\documentclass{article}

\usepackage{arxiv}

\usepackage[utf8]{inputenc} 
\usepackage[T1]{fontenc}    
\usepackage{hyperref}       
\usepackage{url}            
\hypersetup{
    colorlinks = true,
    citecolor = blue,
    urlcolor = blue,
    linkcolor = blue
}
\usepackage{booktabs, multirow}       
\usepackage{amsfonts, amsmath, amssymb, amsthm, bm}       
\usepackage{graphicx, caption, subcaption}  
\usepackage[ruled, vlined]{algorithm2e}   

\usepackage[defernumbers=true, sorting=none, giveninits=true]{biblatex}
\makeatletter
\let\blx@rerun@biber\relax
\makeatother

\theoremstyle{definition}
\newtheorem{theorem}{Theorem}
\newtheorem{lemma}{Lemma}
\newtheorem{corollary}{Corollary}

\newcommand{\R}{\mathbb{R}}
\newcommand{\N}{\mathbb{N}}
\newcommand{\norm}{\mathcal{N}}
\newcommand{\prob}{\mathbb{P}}
\newcommand{\E}{\mathbb{E}}

\title{Rough-Fuzzy CPD:  A Gradual Change Point Detection Algorithm} 

\author{
    Ritwik Bhaduri\\
    Indian Statistical Institute\\
    Kolkata, India\\
    \texttt{ritwik.bhaduri@gmail.com}\\
    \And
    Subhrajyoty Roy\\
    Indian Statistical Institute\\
    Kolkata, India\\
    \texttt{roysubhra98@gmail.com}\\
    \And
    Sankar K. Pal\\
    Center for Soft Computing Research, \\
    Indian Statistical Institute\\
    Kolkata, India\\
    \texttt{sankar@isical.ac.in}\\
}

\begin{document}
\maketitle

\begin{abstract}
Changepoint detection is the problem of finding abrupt or gradual changes in time series data when the distribution of the time series changes significantly. There are many sophisticated statistical algorithms for solving changepoint detection problem, although there is not much work devoted towards gradual changepoints as compared to abrupt ones. Here we present a new approach to solve changepoint detection problem using fuzzy rough set theory which is able to detect such gradual changepoints. An expression for the rough-fuzzy estimate of changepoints is derived along with its mathematical properties concerning fast computation. In a statistical hypothesis testing framework, asymptotic distribution of the proposed statistic on both single and multiple changepoints is derived under null hypothesis enabling multiple changepoint detection. Extensive simulation studies have been performed to investigate how simple crude statistical measures of disparity can be subjected to improve their efficiency in estimation of gradual changepoints. Also, the said rough-fuzzy estimate is robust to signal-to-noise ratio, high degree of fuzziness in true changepoints and also to hyper parameter values. Simulation studies reveal that the proposed method beats other fuzzy methods and also popular crisp methods like WBS, PELT and BOCD in detecting gradual changepoints. The applicability of the estimate is demonstrated using multiple real-life datasets including Covid-19. We have developed the python package \hyperlink{https://pypi.org/project/roufcp/}{\texttt{roufcp}} for broader dissemination of the methods.
\end{abstract}

\keywords{Gradual changepoint detection \and Fuzzy set \and Rough set \and Rough-fuzzy Entropy \and Python package \texttt{roufcp}}


\section{Introduction}

\label{sec:intro}

Changepoint detection (CPD) is the process of detecting changes in the probability distribution of a time series data. There are mainly two different directions of the usabiliy of CPD. First direction concerns with the subject of statistical quality control, where any kind of abrupt changes in the data generating process of a time series observation can be detected through CPD methods, which serves the purpose of alarming the concerning authority ~\cite{lai1995sequential}. On the other hand, various time series modelling rely on CPD techniques to segment the study period, and then apply different modelling for each homogeneous segments~\cite{lemire2007better}.

First CPD method was possibly proposed by Hinkley~\cite{hinkley1970inference} based on the likelihood function and it was able to detect at most one changepoint present in the data. With a clustering approach, Scott and Knott~\cite{scott1974cluster} extended the idea to the detection of multiple changepoints based on Binary Segmentation, which was theoretically corroborated by a similar study of Sen and Srivastava~\cite{sen1975tests}. Picard~\cite{picard1985testing} extended the algorithm to a much general setup, using a generalized version of Kolmogorov-Smirnov statistic, as well as its simplification towards the mean and variance change in an autoregressive process. Some recent and popular algorithms for CPD include Pruned Exact Linear Time (PELT)~\cite{killick2012optimal}, Bayesian Change-point Detection analysis~\cite{turner2009adaptive,adams2007bayesian,knoblauch2018spatio}, Kernel Change-point Detection~\cite{harchaoui2009kernel}, Change-point detection using nonparametric statistical measures~\cite{haynes2017computationally}.

While there remains a vast amount of literature and existing methods on detection of abrupt changepoints~\cite{tartakovsky2014sequential,aminikhanghahi2017survey,van2020evaluation}, very few studies are concerned with gradual changepoints~\cite{chang2015fuzzy}. However, in different domains such as linguistics~\cite{sharma2011cognitive}, paleoclimatology~\cite{trauth2018abrupt}, paleontology~\cite{rose1986gradual}, and paleobiology~\cite{hunt2008gradual}, the changes in different phenomenons are often gradual in nature and not abrupt. In such cases, the changepoint detection algorithms developed to tackle the abrupt changes are found to be inefficient in estimation due to relatively larger variance. However, combining the theory of Fuzzy Rough sets allows an improvement over existing CPD algorithms in terms of increasing the efficiency in estimation, as will be shown in the paper.

The very few investigations focused on detection of fuzzy or gradual changepoints are based on fuzzy clustering techniques. Almost all of them~\cite{chang2015fuzzy,yu2001general} use fuzzy c-means clustering method or its variants to segment the period of study into homogeneous segments. Other approaches include fuzzy regression algorithms~\cite{LU201612,chang2015fuzzy}, and neural networks~\cite{d2011incipient}. No approach considers the discernibility issues that are prevalent in the examples with gradual changepoints, for instance, observations pertaining to two close time points would naturally come from a similar kind of distributions, converting each of the segments into a rough set instead of a crisp one. The present paper aims to provide a general methodology based on fuzzy rough sets to increase the efficiency of any statistical measure that can be used to detect changepoints.

The novelty of this paper is three fold. The first aspect is that we mathematically model uncertainty associated with gradual changepoints as a fuzzy set and the indiscernibility between very close time points using rough set theory. This enables us to devise a changepoint detection algorithm which is relatively robust to both low signal-to-noise ratio of the data and the high degree of fuzziness of the actual changepoint. Simulation studies show that our estimates remain fairly accurate over an extensive range of hyper parameter values furthering the robustness of the proposed method. The second aspect of this paper is that we connect the image segmentation problem from the field of Pattern Recognition and the changepoint detection problem from statistical domain. This connection was discussed by Chatterjee et al.~\cite{equivalence_of_image_segmentation_CP}, where they had formally expressed image segmentation as a special case of changepoint detection problem. However, our work can be interpreted as going to a reverse direction, by generalization of methods proposed for image segmentation by Sen and Pal~\cite{image-ambiguity-skpal} for the problem of changepoint detection. The third component of novelty is reducing the computational complexity of the aforementioned image segmentation algorithm by a series of results, and fitting the algorithm in the context of statistical changepoint detection problem. To formally present the proposed algorithm in conjunction with a statistical hypothesis testing framework, we derive the asymptotic distribution of the proposed rough-fuzzy estimator under some suitable regularity conditions. This enables us to test for false positives of the detected changepoints and also allows for multiple changepoint detection. The performance of the proposed method has been illustrated using extensive simulation studies as well as real data applications including COVID-19 dataset. We observe that our method outperforms other fuzzy CPD methods as well as state of the art crisp CPD methods like Wild Binary Segmentation, PELT and BOCD in detecting gradual changepoints. The implementation of our method is made available through a python package \hyperlink{https://pypi.org/project/roufcp/}{\texttt{roufcp}}, bearing the acronym for the name of the proposed algorithm Rough-Fuzzy CPD.

The rest of the paper is organized as follows: In section \ref{sec:math-background}, we formulate the problem of CPD and provide a brief introduction to fuzzy set, rough set, and the entropy used to define our statistic. Section \ref{sec:proposed-method} provides a general development of the proposed method and defines a rough-fuzzy estimate of change point. Section \ref{sec:math-property} contains the mathematical results of the paper, divided into two subsections. The first subsection deals with the reduction of computational complexity of the proposed method, while the second subsection provides an asymptotic result which can be put into hypothesis testing framework to detect false positives and multiple changepoint detection. In section \ref{sec:simulation}, we perform extensive simulation studies of the performance of the proposed model under different signal-to-noise ratios, levels of fuzziness of actual changepoint and hyper parameter values. We also compare the performance of our algorithm Rough-Fuzzy CPD with other fuzzy as well as crisp methods for different simulation setups. In section \ref{sec:real-examples}, we apply the proposed method on real-life datasets. Section \ref{sec:conclusion} concludes the paper with adequate discussion on strengths and a few limitations of the proposed method, as well as the possibility of future extensions of our work.

\section{Mathematical Backgrounds}\label{sec:math-background}

\subsection{Problem of changepoint Detection}
The mathematical setup of the problem of changepoint detection starts with a multivariate time series signal $\left\{y_t \right\}_{t=1}^{T}$, where each $y_t \in \R^p$. There also exist some fixed points $1 < t_1^\ast < t_2^\ast < \dots t_k^\ast < T$, such that within the interval $\left[ t_{(i-1)}, t_i \right]$, the time series observations follow same probability distributional model, but these models differ from one interval to another. In particular,

\begin{equation}
    y_t \sim \begin{cases}
    F_0(t) & \text{ if } 1 \leq t < t_1\\
    F_1(t) & \text{ if } t_1 \leq t < t_2\\
    \dots & \dots \\
    F_{(k-1)}(t) & \text{ if } t_{(k-1)} \leq t < t_k\\
    F_k(t) & \text{ if } t_k \leq t \leq T\\
    \end{cases}
    \label{eqn:cp-description}
\end{equation}

and $F_i \neq F_{(i+1)}$ for $i = 0, 1, \dots (k-1)$. A special case follows when the signal is assumed to be piece-wise stationary. The problem of CPD deals with identifying the set of points $\left\{ t_1, t_2, \dots t_k\right\}$ where the distribution or behavior of the signal changes, from the knowledge of the available time series observations $y_t$'s. The knowledge of the number of changepoints $k$ may or may not be known a-priori, in which case it also has to be estimated from the data.

A general framework for existing CPD algorithms is provided in a comprehensive review by Truong et al.~\cite{truong2020selective}. Most of the existing algorithms involve minimization of a criterion function $V(\mathcal{T}, y)$ to obtain the best set of changepoints $\mathcal{T}$, 

\begin{equation}
V(\mathcal{T}, y) = \sum_{i = 1}^{k} c\left(\left\{ y_t \right\}_{t_{(i-1)}}^{t_i}\right) + \text{Penalty}(\mathcal{T})
\label{eqn:v-cost}    
\end{equation}

\subsection{Fuzzy Set}
Fuzzy sets, proposed by Lotfi A. Zadeh~\cite{zadeh1965fuzzy}, provide a mathematical means of modelling vagueness and imprecise information. They are a generalization of crisp sets or normal sets. A fuzzy set is represented by an ordered pair $\langle X , \mu \rangle$. Denoting $\mathbb{U}$ as the universe of objects, and $X \subseteq \mathbb{U}$, the membership function for the set $X$ can be described as $\mu : \mathbb{U} \rightarrow [0, 1]$, indicating the degree on inclusion of an element $x \in \mathbb{U}$ into the set $X$. Consequently, 3 cases are possible based on the value of $\mu(x)$:

\begin{enumerate}
\item Not included in $\langle \mathbf{X}, \mu\rangle$ if  $\mu(x) = 0$.
\item Partially included in $\langle \mathbf{X}, \mu\rangle$ if  $0 < \mu(x) < 1$.
\item Fully included in $\langle \mathbf{X}, \mu\rangle$ if  $\mu(x) < 1$.
\end{enumerate}

While fuzzy sets have been extensively used in fields like control systems for electronic devices~\cite{chen2001introduction}, decision-making support systems in businesses~\cite{zimmermann1998fuzzy}, prediction systems in finance~\cite{bojadziev2007fuzzy}, quantitative pattern analysis for industrial quality assurance~\cite{kabir2018review}, and uncertainty modelling in pattern recognition, image processing and vision~\cite{mitra2005fuzzy}, its use in changepoint analysis and statistics, in general, has been very limited. In this paper, we shall model the membership of a data point to a particular distribution as a fuzzy set to characterize the ambiguity in the data generation process.

\subsection{Rough set}
The notion of rough sets, as introduced by Pawlak~\cite{pawlak1982rough}, is defined as follows: Let $\mathcal{A}$ = $\langle \mathbb{U}, A \rangle$ be an information system, where $\mathbb{U}$ is the universe of all objects as mentioned before, and $A$ be the set of attributes which are used to identify two elements of $\mathbb{U}$. Let $B\subseteq A$ be a subset of attributes and $X \subseteq \mathbb{U}$. We can approximate the set $X$ using only the information contained in $B$ by constructing lower and upper approximations of $X$. Let $[x]_B$ denote the equivalence class of object $x$ relative to $I_B$ (equivalence relation induced by the variables in $B$). In rough set theory, the upper and lower approximations of the set $X$ under $I_B$ are defined as follows;

\begin{itemize}
    \item ${\underline{B}X} = $ B-Lower approximation of $X$ in $\mathbb{U}$ =$ \{x \in \mathbb{U}: [x]_B \subseteq X\}$
    \item ${\overline{B}X}$ = B-upper approximation of $X$ in $\mathbb{U}$ = $ \{x \in \mathbb{U}: [x]_B \cap X \neq \phi\}$
\end{itemize}

So, the lower approximation of a set $X$ relative to $B$ in $\mathbb{U}$ are the elements in $\mathbb{U}$ which can be certainly classified as elements of $X$ based on B. Intuitively, it is the set of all elements $x$ in $\mathbb{U}$ such that the equivalence class containing $x$ is a subset of target set $X$. In other words, the lower approximation is the set of objects that are certainly members of the target set $X$.

On the other hand, the upper approximation is the complete set of objects $x$ such that the equivalence class containing $x$ has non empty intersection with $X$. Intuitively, it is the set of elements in $\mathbb{U}/B$ that cannot be positively (i.e., unambiguously) classified as belonging to the complement of $X$ of the target set $X$. In contrast to lower approximation, the upper approximation is the complete set of objects that are certainly as well as possibly members of the target set $X$.

The pair $\langle \underline{B}X, \overline{B}X\rangle$ denotes the rough representation of the crisp set $X$ with respect to $B$. This rough representation actually captures the uncertainty in defining $X$ in view of the incomplete information provided by the subset of attributes $B$. Pawlak~\cite{pawlak1991rough} provided a numerical characterization of roughness of $X$ as 
\begin{equation}\label{eq:rho}
    \rho_B(X) = 1 - \dfrac{|\underline{B}X|}{|\overline{B}X|}
\end{equation}

So $\rho_B = 0$ means, the set $X$ is crisp or exact  (with respect to $B$) and conversely, $\rho_B > 0$ means, $X$ is rough, i.e. ambiguous, with respect to $B$. This vague definition of $X$ in $\mathbb{U}$ (in terms of lower and upper approximations) signifies incompleteness of knowledge about $\mathbb{U}$.

\subsection{Lower and upper approximations of a fuzzy set}
Similar to the lower and upper approximation of a crisp set, definition of such approximations for a fuzzy set can also be established. While the definition of lower and upper approximations of a set may vary based on the nature of set $X$ (crisp or fuzzy) as well as the nature of relation between elements of the universe $\mathbb{U}$, here, we are only concerned with approximations of a fuzzy set $X$ with respect to an indiscernibility relation $R$ which is either an equivalence relation (crisp or fuzzy) or a tolerance relation (crisp or fuzzy). The expressions of upper and lower approximations for both these cases have been derived by Sen and Pal~\cite{image-ambiguity-skpal}. Since the set approximations obtained using equivalence relations are not always smooth, we will restrict our consideration only on the tolerance relations. The general implicit expressions of these approximations are given below, while a more explicit form of such approximations will be derived later in eqn.~\ref{sec:math-property}.

A tolerance relation (crisp or fuzzy) is a (crisp or fuzzy) relation which satisfies (crisp or fuzzy) reflexivity and symmetry. Unlike equivalence relations, tolerance relations are not necessarily transitive. When $R$ is a tolerance relation, the space $\langle \mathbb{U}, R \rangle$ is called tolerance approximation space. Associated with a tolerance relation $R$, there is a membership function $S_R$ of the relation itself, where $S_R(u, v)$ denotes the membership value of the pair $(u, v)$ belonging to the relation $R$. With the help of these, the lower and upper approximations of any crisp or fuzzy set $X$ are obtained as follows;
\begin{equation}
    \begin{split}
        \underline{R}X = \{ (u, \underline{M}(u)) | u \in \mathbb{U}\}\\
    \overline{R}X = \{ (u, \overline{M}(u)) | u \in \mathbb{U}\}
    \end{split}
    \label{eqn:set-approx}
\end{equation}
where the approximations of the membership functions are based on the tolerance function $S_R$ of the relation $R$;
\begin{equation}
    \begin{split}
        \underline{M}(u) &= \inf_{\psi \in \mathbb{U}} \max (1-S_R(u, \psi), \mu_X(\psi))\\
        \overline{M}(u) &= \sup_{\psi \in \mathbb{U}} \min (S_R(u, \psi), \mu_X(\psi))        
    \end{split}
    \label{eqn:set-approx-membership}
\end{equation}
where $\mu_X$, which takes values in the interval $[0, 1]$, is the membership function associated with set $X$. When $X$ is a crisp set, $\mu_X$ would take values only from the set $\{0, 1\}$. Similarly, when $R$ is a crisp tolerance function, $S_R(u, \psi)$ would take values only from the set $\{0, 1\}$.

\subsection{Entropy measures} \label{sec:entropy}
Now, using the definitions of upper and lower approximations of a set, an entropy can be defined to quantify the ambiguity in description of $X$. While the roughness measure $\rho_B(X)$ as in eqn.~\ref{eq:rho} gives a measure of ambiguity in the description of $X$, they are needed to be transformed to properly describe information gain. In this regard, two types of gain functions as defined by Pal and Pal~\cite{entropy-skpal} are considered.

\subsubsection{Logarithmic entropy} 

This is derived by using logarithmic function to measure the gain in incompleteness similar to ``gain in information" in Shanon's entropy. The logarithmic entropy measure for quantifying the incompleteness of knowledge about $\mathbb{U}$ with respect to the definability of a set $X\in \mathbb{U}$ is given as 
\begin{equation}\label{eq:log_entropy}
    H_R^L(X) = -\frac{1}{2}\left(\chi(X) + \chi(X^C) \right)
\end{equation}

where for any set $D \in \mathbb{U}$, $\chi(D) = \rho_R(D)\log_\beta\left(\frac{\rho_R(D)}{\beta}\right)$. Note that the ``gain in incompleteness" term is taken as $\log_\beta\left(\frac{\rho_R(D)}{\beta}\right)$ and for $\beta > 1$ it takes values in $[1,\infty]$.

\subsubsection{Exponential entropy} 

The other kind of entropy measure is defined using exponential function for the ``gain in incompleteness". This class of entropy is derived using $\chi(D) = \rho_R(D)\beta^{\overline{\rho}_R(D)}$ where $\overline{\rho}_R(D) = 1-\rho_R(D)$. So the expression of entropy becomes;
\begin{equation}\label{eq:exp_entropy}
    H_R^E(X) = -\frac{1}{2}\left(\rho_R(X)\beta^{\overline{\rho}_R(X)} + \rho_R(X^C)\beta^{\overline{\rho}_R(X^C)} \right)
\end{equation}

Here the gain in incompleteness term is taken as $\beta^{(1-\rho_R)}$ which takes values in $[1, \beta]$ when $\beta>1$. This class of exponential entropy functions possesses various desirable properties which are not present in the usual Shanon's entropy as illustrated by Pal and Pal~\cite{entropy-skpal}. Some of them are as follows:

\begin{enumerate}
    \item In Shanon's entropy the gain in information $\log(1/p) \rightarrow  
    \infty$ as $p \rightarrow 0$ and is undefined for $p = 0$. However, in real life, gain in information from an event, whether highly unlikely or highly probable is expected to be finite. However, exponential gain function ensures that such gain in information is bounded between $1$ and $\beta$.
    \item Logarithmic entropy is very sensitive to outliers due to the nature of $\log$ function. In contrast, exponential entropy is more robust to outliers. Particularly, a small probability relative to the other probabilities will bias the entropy towards the small probability event since the weight function $\log(1/p)$ in logarithmic entropy has an unbounded derivative. In comparison, the weight function $\beta^{(1-p)}$ used in exponential entropy has bounded derivative.
\end{enumerate}

We have considered $\beta = e$ in the subsequent sections, but any value greater than $1$ is suitable.

\section{Proposed Method: Detecting Rough-fuzzy changepoint}\label{sec:proposed-method}

In terms of mathematical setup, we consider $\left\{ y_1, y_2, \dots y_T  \right\}$ as the time series data with $y_t \in \R^p$ for any $t = 1, 2, \dots T$. With the number of changepoints equal to $1$, we assume that the data comes from a distribution $\mathcal{F}$ at first and then gradually comes from a different distribution $\mathcal{G}$ after some time. Since the location of the changepoint is ambiguous in nature, this creates the possibility of splitting the set of time points $\mathbb{U} = \{1,2, ... ,T\}$ into two fuzzy partitions, $\gamma_\mathcal{F} = \langle\mathbb{U}, \mu_{\mathcal{F}}\rangle$ and $\gamma_\mathcal{G} = \langle\mathbb{U}, \mu_{\mathcal{G}}\rangle$. Here, $\mu_{\mathcal{F}}$ and $\mu_{\mathcal{G}}$ denote the membership function of the respective partitions with observations coming from $\mathcal{F}$ and $\mathcal{G}$ respectively. Clearly, an obvious restriction is that $\mu_{\mathcal{F}}(t) + \mu_{\mathcal{G}}(t) = 1$. This is related to the incompleteness in knowledge about $\mathbb{U}$ and can be quantified by the entropy measures described by Sen and Pal~\cite{image-ambiguity-skpal}.

To formalize this notion of fuzziness, for $t \in \mathbb{U}$, for an estimated changepoint $s$ and bandwidth $\Delta$, we define a fuzzy measure similar to that in~\cite{image-ambiguity-skpal}.
\begin{equation}
    \mu_{s,\Delta}(t) = 
        \begin{cases}
          1 &  t \leq s -\Delta\\
          1-2\left[\dfrac{t-(s -\Delta)}{\Delta}\right]^2 & s -\Delta < t \leq s \\
          2\left[\dfrac{(s+\Delta) - t}{\Delta}\right]^2 & s < t \leq s+\Delta \\
          0 &  t > s+\Delta
        \end{cases}
    \label{eqn:fuzzy-membership}
\end{equation}

On the basis of this estimated changepoint $s$, the partitions $\gamma_{\mathcal{F}}$ and $\gamma_{\mathcal{G}}$ can be reformalized as $\gamma_s = \gamma_{\mathcal{F}} = \{(t, \mu_{s, \Delta}(t)) | t\in \mathbb{U}\}$ and $\gamma_s^C = \gamma_{\mathcal{G}} = \{(t, 1 - \mu_{s, \Delta}(t)) | t\in \mathbb{U}\}$, which captures the fuzzy nature of the two partitions due to the estimation of changepoint by a fixed quantity $s$.

Following the footsteps of grayness ambiguity, as formulated in~\cite{image-ambiguity-skpal}, we consider a tolerance function $S_w(u, v)$ such that, 

\begin{enumerate}
    \item $S_w(u, u) = 1$ for any $u \in \mathbb{U} = \{ 1, 2, \dots T\}$.
    \item $S_w(u, v)$ is a decreasing function in $\vert u - v\vert$.
    \item $S_w(u, v) = 0$ if $\vert u - v \vert \geq 2w$, where $w$ is a chosen window length. This means, sufficiently spaced timepoints can be distinguished quite nicely.
\end{enumerate}

Then, the lower and upper approximations of the fuzzy set $\gamma_s$ can be constructed as $\underline{\gamma_s} = \left\{ (u, {M}_{\underline{\gamma_s}}) : u \in \mathbb{U} \right\}$ and $\overline{\gamma_s} = \left\{ (u, {M}_{\overline{\gamma_s}}) : u \in \mathbb{U} \right\}$, where ${M}_{\underline{\gamma_s}}$ and ${M}_{\overline{\gamma_s}}$ are obtained using eqn.~\ref{eqn:set-approx-membership} with the tolerance function $S_w$ and membership function $\mu_{s, \Delta}(t)$. In a similar way, lower and upper approximations of the fuzzy complement set $\gamma_s^C$ can also be obtained. While these approximations can be pre-computed without the knowledge of time series data, we incorporate the knowledge of the available data by means of a regularity measure and combine it with the set approximations to obtain an ambiguity measure.

In order to create a regularity measure, for each timepoint $t$, a two sample test statistic is computed to detect the changes in the samples $\{ y_{(t-\delta+1)}, y_{(t-\delta+2)}, \dots y_{t} \}$ and $\{ y_{(t+1)}, y_{(t+2)}, \dots y_{(t+\delta)} \}$. As the regularity measure $R(t)$ is expected to take higher value when there is no change in distribution and take lower value when there is a detected change in distribution, some transformation of test statistic could be used. One such regularity measure to detect the changes in mean could be based on Hotelling's T$^2$ test statistic.

\begin{equation*}
    R(t) = \dfrac{1}{1 + (\bar{y}_{1} - \bar{y}_2)^{\intercal} \Sigma^{-1} (\bar{y}_{1} - \bar{y}_2) }
\end{equation*}
where,
\begin{align*}
    \bar{y}_1 & = \delta^{-1} \hspace{-0.5cm}\sum_{t' = (t-\delta+1)}^{t} \hspace{-0.4cm} y_{t'} \hspace{2cm}
    \bar{y}_2 = \delta^{-1} \hspace{-0.35cm}\sum_{t' = (t+1)}^{(t+\delta)} \hspace{-0.25cm} y_{t'}\\
    \Sigma & = \delta^{-1} \hspace{-0.5cm}\sum_{t' = (t - \delta+1)}^{(t + \delta)} \hspace{-0.2cm} \left( y_{t'} - \dfrac{1}{2}(\bar{y_1} + \bar{y}_2) \right) \left( y_{t'} - \dfrac{1}{2}(\bar{y_1} + \bar{y}_2) \right)^\intercal
\end{align*}

A roughness measure is created for the two fuzzy partitions $\gamma_s$ and $\gamma_s^C$, incorporating the information of changes in the data as provided by suitably transformed test statistic $R(t)$, and combining them with the respective upper and lower approximations of the fuzzy membership functions.
\begin{equation}
    \begin{split}
        \rho_{\Delta, \delta, w}(\gamma_s) & = 
        1 - \dfrac{\sum_{t = 1}^{T} {M}_{\underline{\gamma_s}}(t) R(t) }{\sum_{t = 1}^{T} {M}_{\overline{\gamma_s}}(t) R(t)}\\
        \rho_{\Delta, \delta, w}(\gamma_s^c) & = 
        1 - \dfrac{\sum_{t = 1}^{T} {M}_{\underline{\gamma_s^c}}(t) R(t) }{\sum_{t = 1}^{T} {M}_{\overline{\gamma_s^c}}(t) R(t)}\\        
    \end{split}
    \label{eqn:new-roughness}
\end{equation}

Based on these roughness measures, entropy quantifying the ambiguity for the fuzzy partitions of the time span for the specifically chosen changepoint $s$ can be expressed using eqn.~\ref{eq:log_entropy} or eqn.~\ref{eq:exp_entropy}, with $\rho_R(X)$ replaced by the roughness measure given in eqn.~\ref{eqn:new-roughness}. Thus, we obtain
\begin{equation}
    H^E_{\Delta, \delta, w}(s) = \rho_{\Delta, \delta, w}(\gamma_s) e^{\left(1 - \rho_{\Delta, \delta, w}(\gamma_s)\right)} + \rho_{\Delta, \delta, w}(\gamma_s^C)e^{\left( 1 - \rho_{\Delta, \delta, w}(\gamma_s^C) \right)} 
    \label{eqn:new-entropy}
\end{equation}

Proceeding in the direction of any general changepoint detection method~\cite{truong2020selective} for detecting a single changepoint present in the data, any CPD algorithm can be broadly expressed as
\begin{equation}
    t^\ast = \min_{t \in \{ 1, 2, \dots T\} } R(t) = \min_{t \in \{ 1, 2, \dots T\} } V\left(\left( \{ 1, 2, \dots t \}, \{ (t+1), (t+2), \dots T \} \right), y\right).
    \label{eqn:old-cp}
\end{equation}

where $t^\ast$ is the estimated changepoint and $V(\cdot, \cdot)$ is a cost function as shown in eqn.~\ref{eqn:v-cost}, which can be interpreted as a regularity measure. While in this way, the regularity measure $R(t)$ itself becomes an indicator of the changepoint, it can be greatly enhanced in combination of fuzzy and rough set theory, by constructing the entropy as given in eqn.~\ref{eqn:new-entropy}. Thus, according to the proposed algorithm, the estimated changepoint is given as
\begin{equation}
    t^\ast = \min_{t \in \{ 1, 2, \dots T\} }  H^E_{\Delta, \delta, w}(t).
    \label{eqn:new-cp}
\end{equation}
This estimated changepoint $t^\ast$ shown in eqn.~\ref{eqn:new-cp} may be denoted as the rough-fuzzy CP, and the method may be called rough-fuzzy CPD. Note that while the global minima serves as the estimate of single changepoint, the local minima of the entropy function, after suitable testing, are used for multiple changepoint detection. 

The reason for including roughness in the proposed changepoint detection algorithm is subtle. Since the underlying distributions of the time series data is unknown, all the information about the changepoint must be gathered in terms of the available observations $y_t$, which is now summarized only through a single attribute $R(t)$, the regularity measure. Along with this loss of information, since the regularity measure $R(t)$ is computed based on overlapping windows, any information about the locality of a timepoint will permeate to its neighbouring timepoints as well, resulting in incomplete information about the timepoints itself. This rough resemblance between different time points is modelled by the tolerance relation which eventually leads to the rough set formulation.

Further, the entropy function depends on the hyper parameters $w$, $\delta$ and $\Delta$. In the proposed rough-fuzzy CPD, $w$ denotes the degree of roughness of the tolerance function with higher values indicating greater roughness and $\Delta$ determines the fuzziness of the membership function with higher values corresponding to greater fuzziness. It is important to choose the values of $w$ and $\Delta$ correctly as wrongly chosen values might result in relatively higher error in estimation, though a wide range of such optimal values are available as shown in section~\ref{sec:sensitivity analysis}. 

\section{Mathematical Properties}\label{sec:math-property}
Here we provide two mathematical properties of the aforesaid estimate $t^\ast$, involving rough-fuzzy entropy, towards the detection of gradual changepoints. First one deals with the issue of its speedy computation, while the other theoretically establishes its ability in detecting multiple changepoints. In the section~\ref{sec:fast comuputation}, we obtain exact solutions of  upper and lower approximations and also derive relations between them to reduce computation. In section~\ref{sec:asymptotics}, we prove a theorem which gives the asymptotic distribution of rough-fuzzy entropy under the null hypothesis that there are no changepoints. An immediate corollary of this theorem gives us the joint asymptotic distribution of rough-fuzzy entropy evaluated at multiple proposed changepoints under the null hypothesis. This helps to present our algorithm in a hypothesis testing framework and also test for false positives.

\subsection{Fast Computation of Upper and Lower Approximations}\label{sec:fast comuputation}

We start by introducing a result which relates the upper and lower approximations of the set $\gamma_s$ with the approximations for the complement set $\gamma_s^C$. The significance of the result is that in order to obtain the roughness measure $\rho_{\Delta, \delta, w}(\gamma_s)$ and $\rho_{\Delta, \delta, w}(\gamma_s^C)$ as given in eqn.~\ref{eqn:new-roughness}, it is enough to focus the computation on only one of these terms, and the other can be obtained as a byproduct of the result.

\begin{lemma}
    \label{thm:approx-complement}
    \begin{align*}
        M_{\underline{\gamma}_s}(t) & = 1 - M_{\overline{\gamma}_s^C}(t) \\
        M_{\overline{\gamma}_s}(t) & = 1 - M_{\underline{\gamma}_s^C}(t)\\
    \end{align*}
\end{lemma}

\begin{proof}
    Starting with the lower approximation;
    \begin{align*}
        M_{\underline{\gamma}_s}(t)
        & = \inf_{\psi \in \mathbb{U}} \max\left( \overline{S}_w(t, \psi), \mu_{s, \Delta}(\psi) \right)\\
        & = \inf_{\psi \in \mathbb{U}} \left[ 1 - \min\left( 1 - \overline{S}_w(t, \psi), 1 - \mu_{s, \Delta}(\psi) \right)\right] \\
        & = \inf_{\psi \in \mathbb{U}} \left[ 1 - \min\left( S_w(t, \psi), \overline{\mu}_{s, \Delta}(\psi) \right)\right] \\
        & = 1 - \sup_{\psi \in \mathbb{U}} \min\left( S_w(t, \psi), \overline{\mu}_{s, \Delta}(\psi) \right) \\
        & = 1 - M_{\overline{\gamma}_s^C}(t)
    \end{align*}

    The other equality proceeds in an exactly similar way.
\end{proof}

Also, since the computation of the lower and upper approximations ${M}_{\underline{\gamma_s}}(t)$ and ${M}_{\overline{\gamma_s}}(t)$ is independent of the data, it can be pre-computed for the changepoint analysis problems, given the knowledge of the number of timepoints $T$. However, if $T$ is large, computation of eqn.~\ref{eqn:set-approx-membership} poses a high memory and computational complexity. However, it is possible to obtain exact expressions of these lower and upper approximations under a very general setup, which greatly reduces both the computational and storage cost complexities of the whole process.

To establish an explicit formula for lower and upper approximations, we first consider the situation when $M_{\underline{\gamma}_s}(t) = 0$, which happens if and only if there exists a $\psi \in \mathbb{U}$ such that, both $\overline{S}_w(t, \psi)$ and $\mu_{s, \Delta}(\psi)$ are equal to $0$. While the complementary tolerance function $\overline{S}_w$ is $0$, if and only if $S_w$ is $1$, i.e. the two arguments satisfy $t = \psi$. On the other hand, $\mu_{s, \Delta}(\psi) = 0$ if and only if $\psi \geq (s+\Delta)$, combining this with $t = \psi$ yields, $t \leq (s + \Delta)$.

On the other extreme, $M_{\underline{\gamma}_s}(t) = 1$, if and only if, either $S_w(t, \psi) = 0$ or $\mu_{\Delta, s}(\psi) = 1$. The former happens when $\vert t - \psi \vert \geq 2w$, and the latter happens if $\psi \leq (s - \Delta)$. Thus, for any $t \leq (s - 2w - \Delta)$, $M_{\underline{\gamma}_s}(t) = 1$.

By symmetry, $M_{\overline{\gamma}_s}(t) = 0$ if $t \geq (s + 2w + \Delta)$ and $M_{\overline{\gamma}_s}(t) = 1$ if $t \leq (s - \Delta)$.

\begin{figure}[ht]
    \centering
    \includegraphics[width = \linewidth, trim = {2cm 5cm 4cm 5cm}, clip]{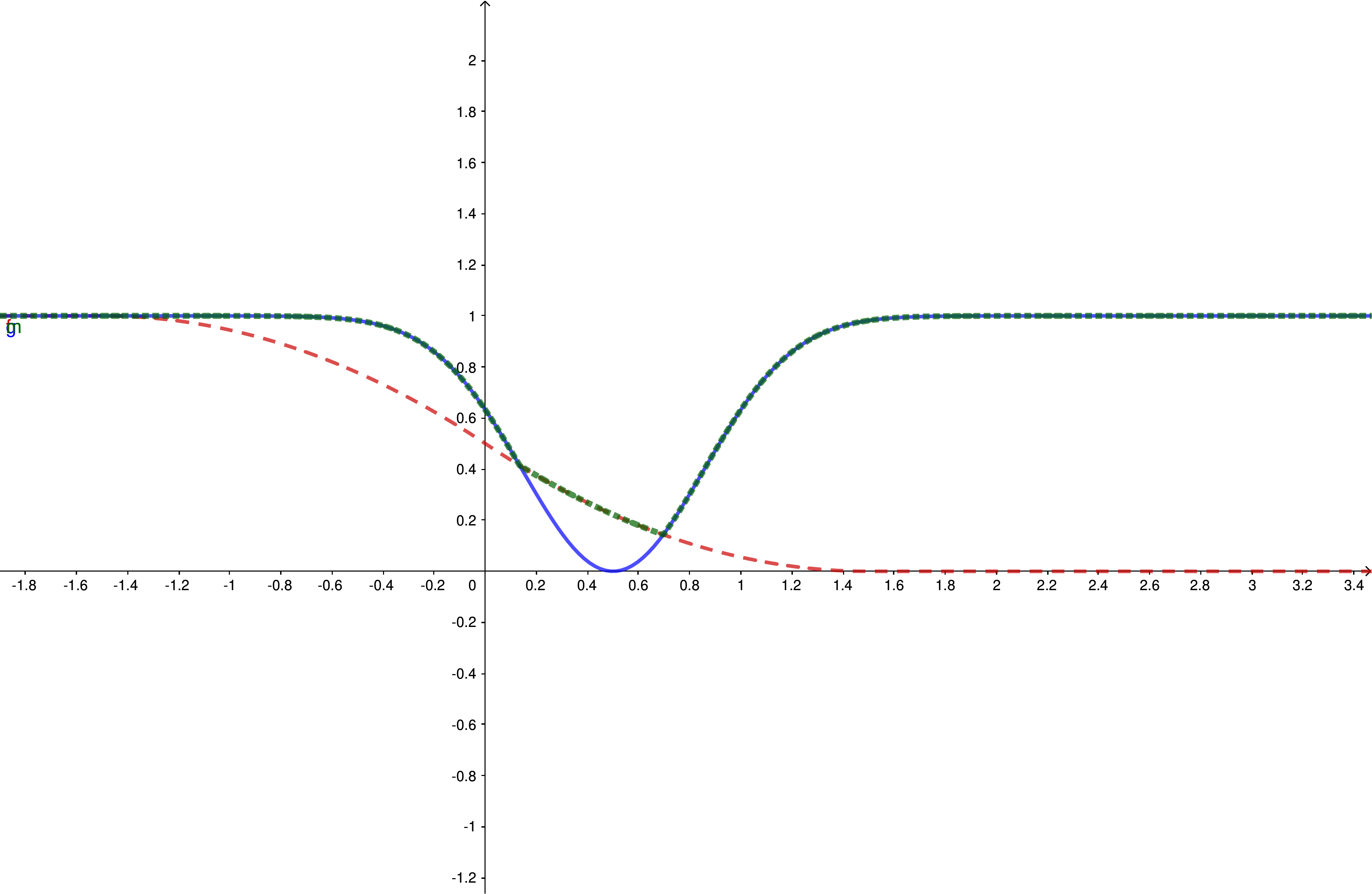}
    \caption{Lower Approximation Curve: red line is $\mu_{T, \Delta}(\psi)$ and blue line is $\overline{S}_w(t, \psi)$} and green line is their maximum.
    \label{fig:lower_approx}
\end{figure}

To see the behaviour of $M_{\underline{\gamma}_s}(t)$ when $(s-2w -\Delta) \leq t \leq (s +\Delta)$, we consider $M_{\underline{\gamma}_s}(t)$ as a minimizer of the function $\max\left( \overline{S}_w(t, \psi), \mu_{s, \Delta}(\psi) \right)$ with respect to $\psi$. This is shown by the upper envelope curve of $\overline{S}_w(t, \psi)$ and $\mu_{s, \Delta}(\psi)$ in figure~\ref{fig:lower_approx}. Thus, it is easy to see that the minimizer would appear at a point $t^\ast$ where the fuzzy membership function and the tolerance function crosses each other, i.e., $\overline{S}_w(t, t^\ast) = \mu_{s, \Delta}(t^\ast)$. This leads us to a result depicting a much easier way to compute the lower and upper approximations.

\begin{theorem}
    \label{thm:exact-upper-lower}
    With the membership function $\mu_{s, \Delta}(t)$, given as in eqn.~\ref{eqn:fuzzy-membership}, the lower and upper approximations of the left partition of a chosen changepoint $s$ are expressed by $\underline{\gamma_s} = \left\{ (u, {M}_{\underline{\gamma_s}}) : u \in \mathbb{U} \right\}$ and $\overline{\gamma_s} = \left\{ (u, {M}_{\overline{\gamma_s}}) : u \in \mathbb{U} \right\}$ respectively, where
    
    \begin{equation*}
        M_{\underline{\gamma_s}}(t) = \begin{cases}
            1 & \text{ if } t < (s - 2w - \Delta)\\
            1 - \displaystyle\max_{ \{ t^\ast : S_w(t, t^\ast) + \mu_{s, \Delta}(t^\ast) = 1 \} } S_w(t, t^\ast) & \text{ if } (s - 2w - \Delta) \leq t < (s + \Delta)\\
            0 & \text{ if } t \geq (s + \Delta)\\
        \end{cases}
    \end{equation*}
    
    and 
    
    \begin{equation*}
        M_{\overline{\gamma_s}}(t) = \begin{cases}
            1 & \text{ if } t < (s - \Delta)\\
            \displaystyle\max_{ \{ t^\ast : S_w(t, t^\ast) = \mu_{s, \Delta}(t^\ast) \} } S_w(t, t^\ast) & \text{ if } (s - \Delta) \leq t < (s + 2w + \Delta)\\
            0 & \text{ if } t \geq (s + 2w + \Delta)\\
        \end{cases}
    \end{equation*}
\end{theorem}

It is possible to obtain a closed form solution of $S_w(t, t^\ast) = \mu_{s, \Delta}(t^\ast)$ and similar equations for some specific tolerance functions. One such specific choice is provided in the following corollary.

\begin{corollary}
    Let the membership function $\mu_{s, \Delta}(t)$ be given as in eqn.~\ref{eqn:fuzzy-membership}, and a tolerance function be given as
    \begin{equation*}
        S_w(t, t') = \begin{cases}
            0 & \text{ if } \vert t - t' \vert \geq 2w\\
            2\left[ \dfrac{t' - (t - 2w)}{2w} \right]^2 & \text{ if } (t-2w) < x < (t-w)\\
            1 - 2\left[ \dfrac{x-t}{2w} \right]^2 & \text{ if } \vert t - t' \vert \leq w\\
            2\left[ \dfrac{(t + 2w) - t'}{2w} \right]^2 & \text{ if } (t+w) < x < (t+2w).\\
        \end{cases}
    \end{equation*}
    Then the lower and upper approximations of the left and right partitions for a chosen changepoint $s$ can be obtained as $\underline{\gamma_s} = \left\{ (u, {M}_{\underline{\gamma_s}}) : u \in \mathbb{U} \right\}$,  $\overline{\gamma_s} = \left\{ (u, {M}_{\overline{\gamma_s}}) : u \in \mathbb{U} \right\}$, $\underline{\gamma_s^C} = \left\{ (u, {M}_{\underline{\gamma_s^C}}) : u \in \mathbb{U} \right\}$ and $\overline{\gamma_s^C} = \left\{ (u, {M}_{\overline{\gamma_s^C}}) : u \in \mathbb{U} \right\}$, where
    
    \begin{equation*}
        M_{\underline{\gamma}_s}(t) = \begin{cases}
            0 & \text{if } t \geq (s + \Delta)\\
            2\left[ \dfrac{(s+\Delta) - t}{2(w + \Delta)} \right]^2 & \text{if }   (s - w) \leq t < (s + \Delta) \\
            1 - 2\left[ \dfrac{(t + 2w) - (s - \Delta)}{2(w + \Delta)} \right]^2 & \text{if } (s - 2w - \Delta) \leq t < (s - w)\\
            1 & \text{if } t < (s - 2w - \Delta) 
        \end{cases}        
    \end{equation*}

    and
    
    \begin{equation*}
        M_{\overline{\gamma}_s}(t) = \begin{cases}
            0 & \text{if } t \geq (s + 2w + \Delta)\\
            2\left[ \dfrac{(s+\Delta) - (t - 2w)}{2(w + \Delta)} \right]^2 & \text{if } (s + w) \leq t < (s + 2w + \Delta)\\
            1 - 2\left[ \dfrac{t  - (s - \Delta)}{2(w + \Delta)} \right]^2 & \text{if } (s - \Delta) \leq t < (s + w)\\
            1 & \text{if } t < (s - \Delta), 
        \end{cases}    
    \end{equation*}
    
    while the approximations for the complementary set can be obtained using lemma~\ref{thm:approx-complement}.
\end{corollary}

\begin{lemma}
    \label{thm:approx-symmetric}
    If the tolerance function $S_{w}(u, v)$ can be expressed as a function of the absolute difference of its arguments, i.e. $S_w(u, v) = g(\vert u - v\vert)$ such that $g(\cdot)$ is symmetric about $0$, then;
    \begin{align*}
        M_{\underline{\gamma}_s}(t) & = 1 - M_{\overline{\gamma}_s}(2s-t) \ \ \forall\  t\ :\ \max\{1, 2s-T\}<t<\max\{2s, T\} 
    \end{align*}
\end{lemma}
\begin{proof}
    We can rewrite theorem~\ref{thm:exact-upper-lower}, as follows, 
    \begin{equation*}
        M_{\overline{\gamma_s}}(t) = \begin{cases}
            1 & \text{ if } t < (s - \Delta)\\
            \displaystyle\max_{ \{ t^\ast : S_w(t, t^\ast) = \mu_{s, \Delta}(t^\ast) \} } \mu_{s, \Delta}(t^\ast) & \text{ if } (s - \Delta) \leq t < (s + 2w + \Delta)\\
            0 & \text{ if } t \geq (s + 2w + \Delta)\\
        \end{cases}
    \end{equation*}
    To prove this lemma, we consider the 3 cases separately.
    \begin{enumerate}
        \item $t < s-2w-\Delta$
        \item $t \geq s+\Delta$
        \item $t \in [s-2w-\Delta,\ s+\Delta)$
    \end{enumerate}
    
    \begin{enumerate}
        \item \emph{Case 1: $t < s-2w-\Delta$}
        Here, $M_{\underline{\gamma}_s}(t)=1$ and $(2s-T) > (s+\Delta+2w)$. Thus, $M_{\overline{\gamma}_s}(2s-t) = 0 = 1 - M_{\underline{\gamma}_s}(t)$ which proves lemma~\ref{thm:approx-symmetric} for this case.
        \item \emph{Case 2: $t \geq s+\Delta$}
        Here, $M_{\underline{\gamma}_s}(t)=0$ and $(2s-T) \leq (s-\Delta)$. Again, symmetric to the previous case, $M_{\overline{\gamma}_s}(2s-t) = 1 = 1 - M_{\underline{\gamma}_s}(t)$ which proves lemma~\ref{thm:approx-symmetric} for case 2 also.
        \item \emph{Case 3: $(s - \Delta) \leq t < (s + 2w + \Delta)$}
        In this case, $(2s-T) \in (s-\Delta,\ s+\Delta+2w]$. Hence,
        \begin{align*}
            M_{\overline{\gamma}_s}(2s-t) &= \max_{ \{ t^\ast : S_w(2s-t, t^\ast) = \mu_{s, \Delta}(t^\ast) \} } \mu_{s, \Delta}(t^\ast)\\
            &=\max_{ \{ t^\ast : S_w(2s-t, t^\ast) = 1-\mu_{s, \Delta}(2s-t^\ast) \} } \left(1-\mu_{s, \Delta}(2s-t^\ast)\right)\\
            & \qquad \qquad \qquad \qquad \qquad \qquad \text{Since, }\mu_{s, \Delta}(t^\ast)=1- \mu_{s, \Delta}(2s-t^\ast)\\
            &=\max_{ \{ t^\ast : S_w(2s-t, t^\ast) + \mu_{s, \Delta}(2s-t^\ast)=1 \} } S_w(2s-t, t^\ast)\\
            &= \max_{ \{ z^\ast : S_w(2s-t, 2s-z^\ast) + \mu_{s, \Delta}(z^\ast)=1 \} } S_w(2s-t, 2s-z^\ast)\\
            & \qquad \qquad \text{Putting } z^\ast=2s-t^\ast\\
            &= \max_{ \{ z^\ast : S_w(t, z^\ast) + \mu_{s, \Delta}(z^\ast)=1 \} } S_w(t, z^\ast)\\
            & \qquad \qquad \qquad \qquad \qquad \text{Since, } S_w(\cdot, \cdot)\text{ is location invariant and }\\
            & \qquad \qquad \qquad \qquad \qquad g(\cdot) \text{ is symmetric about } 0\\
            &= 1- M_{\underline{\gamma}_s}(t)
        \end{align*}    
        So, lemma~\ref{thm:approx-symmetric} also holds for case 3, proving that the result is true in general.
    \end{enumerate}
\end{proof}
The significance of lemma~\ref{thm:approx-symmetric} is that, in case of symmetric and location invariant tolerance function, only one, either lower or upper, approximation curve is required to be computed.

\subsection{Asymptotic distribution}\label{sec:asymptotics}

While the rough-fuzzy CPD can be employed and a single changepoint can be detected using eqn.~\ref{eqn:new-cp}, multiple changepoints can be detected by local minima of the curve $H^E_{\Delta, \delta, w}(s)$ as a function of $s$. However, a reference curve must be computed in order to select the true changepoints from many local minimas. In statistical language, these reference curve is usually computed based on the distribution of the statistic under a suitably chosen null hypothesis, by modifying the problem into a hypothesis testing framework.

Considering the mathematical framework given in eqn.~\ref{eqn:cp-description}, we can formulate the problem of detecting changepoint as a hypothesis testing problem.
\begin{align*}
    H_0: & \quad F_0 = F_1 = \dots = F_k = F\\
    H_1: & \quad \text{There is at least one inequality}
\end{align*}

Since the regularity measure $R(t)$ is an indicator of a possible changepoint, hence under the null hypothesis $H_0$; $\E(R(t))$ is a constant independent of the time $t$. On the basis of this, we obtain asymptotic null distribution of the proposed statistic under some reasonable assumptions on the asymptotic null distribution of the regularity measure. Clearly, the asymptotics follow when the number of samples for constructing $R(t)$ i.e. $\delta$ is tended to infinity, which forces the total number of timepoints $T$ to be tended to infinity as well. Thus, in order to talk about asymptotics, we restrict ourselves to an infinite dimensional normed space, which without any loss of generality, can be taken as $l^{\infty}(\N)$, the set of all uniformly bounded functions from $\N$ to $\R$. Also, while the number of timepoints $T$ increases to infinity, the intervals between successive observations tend to zero, in order to make the total period of observation constant pertaining to most of the practical situations. 

Let us denote the regularity measure $R_{\delta_n}$ as a $l^{\infty}(\N)$ valued random element defined on some probability measure space $(\Omega, \mathcal{A}, \prob)$ such that $R_{\delta_n}(t)$ is a $\R$-valued random variable denoting the regularity measure based on the $2\delta_n$ subsamples centered at $t$ for each $t = 1, 2, \dots \infty$. 

\begin{theorem}
    \label{thm:delta-method-single}
    Assume that under the null hypothesis that there is no changepoint, the regularity measure $R_{\delta_n}$ has an asymptotic distribution such that for some sequence $a_{\delta_n} \rightarrow \infty$, 
    \begin{equation*}
        a_{\delta_n}(R_{\delta_n} - \mu \bm{1}(\cdot)) \rightsquigarrow Z
    \end{equation*}
    where $\bm{1} \in l^{\infty}(\N)$ is the identity function $\bm{1}(x) = x$, $Z$ is an infinite dimensional Gaussian process with mean function identically equal to $0$ and covariance function $\sigma : \N \times \N \rightarrow [0, \infty)$.
    
    Also, assume the following regularity conditions:
    
    \begin{enumerate}
        \item $\mu \neq 0$.
        \item For the proposed changepoint $s$, the series $\sum_{t = 1}^{\infty} M_{\overline{\gamma_s}}(t)$ and $\sum_{t = 1}^{\infty} M_{\overline{\gamma_s^C}}(t)$ are convergent.
    \end{enumerate}
    
    Define, 
        \begin{align*}
            b(s) & = \dfrac{\sum_{t = 1}^{\infty} M_{\underline{\gamma_s}}(t) }{\sum_{t = 1}^{\infty} M_{\overline{\gamma_s}}(t) }\\
            \overline{b}(s) & = \dfrac{\sum_{t = 1}^{\infty} M_{\underline{\gamma_s^C}}(t) }{\sum_{t = 1}^{\infty} M_{\overline{\gamma_s^C}}(t) }
        \end{align*}
    
    Then, the exponential entropy based statistic $H^E_{\Delta, \delta_n, w}(s)$ defined in eqn.~\ref{eqn:new-entropy} has an asymptotic distribution with the same normalizing constant, $a_{\delta_n}$ such that as $a_{\delta_n} \rightarrow \infty$,
    \begin{equation*}
        a_{\delta_n}(H^E_{\Delta, \delta_n, w}(s) - H^\ast(s)) \rightsquigarrow Z^\ast
    \end{equation*}
    where 
    \begin{equation*}
        H^\ast(s) = \left( 1 - b(s) \right)e^{b(s)} + (1 - \overline{b}(s))e^{\overline{b}(s)}
    \end{equation*}
    and $Z^\ast$ is a univariate normally distributed random variable with mean $0$ and variance $\sigma^\ast$ where;
    
    \begin{equation*}
        \sigma^\ast = \sum_{m = 1}^\infty \sum_{n = 1}^\infty A_s(m)\sigma(m, n)A_s(n)
    \end{equation*}
    
    and;
    \begin{equation}
        \begin{split}
            A_s(n) 
            & = \left[b(s)e^{b(s)}
            \left\{ 
            \dfrac{\sum_{t = 1}^{\infty} M_{\overline{\gamma_s}}(n) M_{\underline{\gamma_s}}(t) - \sum_{t = 1}^{\infty} M_{\underline{\gamma_s}}(n) M_{\overline{\gamma_s}}(t) }{\mu \left(\sum_{t = 1}^{\infty} M_{\overline{\gamma_s}}(t)\right)^2}
            \right\} + \right. \\
            & \qquad \qquad \qquad \qquad
            \left. \overline{b}(s)e^{\overline{b}(s)}
            \left\{ 
            \dfrac{\sum_{t = 1}^{\infty} M_{\overline{\gamma_s^C}}(n) M_{\underline{\gamma_s^C}}(t) - \sum_{t = 1}^{\infty} M_{\underline{\gamma_s^C}}(n) M_{\overline{\gamma_s^C}}(t) }{\mu \left(\sum_{t = 1}^{\infty} M_{\overline{\gamma_s^C}}(t)\right)^2}
            \right\}
            \right]
        \end{split}
        \label{eqnthm:A-function}
    \end{equation}
    
    provided that the series expression of $\sigma^\ast$ is convergent.
\end{theorem}

\begin{proof}
    We first fix the parameters $\Delta, w$ and fix the proposed changepoint $s$. Let us define, $\rho_{\Delta, \delta_n, w}(\gamma_s)$ as a function from $l^\infty(\N)$ to $\R$, denoted by $\rho^{(s)}$;
    \begin{equation}
        \rho^{(s)}(R_{\delta_n}) = 1 - \dfrac{\sum_{t = 1}^{\infty} M_{\underline{\gamma_s}}(t)R(t) }{\sum_{t = 1}^{\infty} M_{\overline{\gamma_s}}(t)R(t) }
        \label{eqnthm:roughness-function}
    \end{equation}
    
    We would first compute the Fr\'echet derivative of this at $\mu \bm{1}(\cdot)$. 
    \begin{align*}
        & \left\vert \rho^{(s)}(\mu \bm{1} + h) - \rho^{(s)}(\mu \bm{1})  - \dfrac{\sum_{n = 1}^{\infty}\left[\sum_{t = 1}^{\infty} M_{\overline{\gamma_s}}(n) M_{\underline{\gamma_s}}(t) - \sum_{t = 1}^{\infty} M_{\underline{\gamma_s}}(n) M_{\overline{\gamma_s}}(t) \right]h(n) }{\mu \left(\sum_{t = 1}^{\infty} M_{\overline{\gamma_s}}(t)\right)^2 } \right\vert\\
        = \quad &  \left\vert \dfrac{\sum_{n = 1}^{\infty}\left[\sum_{t = 1}^{\infty} M_{\overline{\gamma_s}}(n) M_{\underline{\gamma_s}}(t) - \sum_{t = 1}^{\infty} M_{\underline{\gamma_s}}(n) M_{\overline{\gamma_s}}(t) \right]h(n) }{\mu \left(\sum_{t = 1}^{\infty} M_{\overline{\gamma_s}}(t) \right)^2 + \left(\sum_{t = 1}^{\infty} M_{\overline{\gamma_s}}(t) \right)\left(\sum_{t = 1}^{\infty} M_{\overline{\gamma_s}}(t) h(t)\right) } \right.\\
        & \qquad \qquad \qquad \qquad - \left. \dfrac{\sum_{n = 1}^{\infty}\left[\sum_{t = 1}^{\infty} M_{\overline{\gamma_s}}(n) M_{\underline{\gamma_s}}(t) - \sum_{t = 1}^{\infty} M_{\underline{\gamma_s}}(n) M_{\overline{\gamma_s}}(t) \right]h(n) }{\mu \left(\sum_{t = 1}^{\infty} M_{\overline{\gamma_s}}(t) \right)^2 } \right\vert\\
        = \quad & \left\vert\dfrac{ \sum_{n = 1}^{\infty}\left[\sum_{t = 1}^{\infty} M_{\overline{\gamma_s}}(n) M_{\underline{\gamma_s}}(t) - \sum_{t = 1}^{\infty} M_{\underline{\gamma_s}}(n) M_{\overline{\gamma_s}}(t) \right]h(n) }{\mu \left(\sum_{t = 1}^{\infty} M_{\overline{\gamma_s}}(t) \right)^2 } \right\vert \left\vert \dfrac{\sum_{t = 1}^{\infty} M_{\overline{\gamma_s}}(t) h(t)}{ \sum_{t = 1}^{\infty} M_{\overline{\gamma_s}}(t)  (\mu + h(t))} \right\vert\\
    \end{align*}
    where $\vert \cdot \vert$ represents the usual absolute value norm on real numbers. Now, since $\mu \neq 0$, without loss of generality assume that, $\mu > \epsilon > 0$ for some small non-negative $\epsilon$. With $\Vert h\Vert_{\infty} \rightarrow 0$, we can thus make,
    \begin{equation*}
        \left\vert \dfrac{\sum_{t = 1}^{\infty} M_{\overline{\gamma_s}}(t) h(t)}{ \sum_{t = 1}^{\infty} M_{\overline{\gamma_s}}(t)  (\mu + h(t))} \right\vert \leq \left\vert \dfrac{\sum_{t = 1}^{\infty} M_{\overline{\gamma_s}}(t) h(t)}{ \sum_{t = 1}^{\infty} M_{\overline{\gamma_s}}(t)  (\epsilon / 2)} \right\vert 
        \leq B \Vert h\Vert_\infty
    \end{equation*}
    
    for sufficiently small $\Vert h \Vert_\infty$ and some finite real number $B$. On the other hand, 
    \begin{align*}
        & \left\vert \dfrac{\sum_{n=1}^\infty \left[\sum_{t = 1}^{\infty} M_{\overline{\gamma_s}}(n) M_{\underline{\gamma_s}}(t) - \sum_{t = 1}^{\infty} M_{\underline{\gamma_s}}(n) M_{\overline{\gamma_s}}(t) \right] h(n)}{\mu \left(\sum_{t = 1}^{\infty} M_{\overline{\gamma_s}}(t)\right)^2} \right\vert \\
        \leq \quad & \sum_{n=1}^\infty \left\vert \dfrac{\left[\sum_{t = 1}^{\infty} M_{\overline{\gamma_s}}(n) M_{\underline{\gamma_s}}(t) - \sum_{t = 1}^{\infty} M_{\underline{\gamma_s}}(n) M_{\overline{\gamma_s}}(t) \right] h(n)}{\mu \left(\sum_{t = 1}^{\infty} M_{\overline{\gamma_s}}(t)\right)^2} \right\vert \qquad \text{by triangle inequality}\\
        \leq \quad & \dfrac{2\Vert h\Vert_\infty}{\mu \left(\sum_{t = 1}^{\infty} M_{\overline{\gamma_s}}(t)\right)^2} \sum_{n=1}^\infty \left\vert \sum_{t = 1}^{\infty} M_{\overline{\gamma_s}}(n) M_{\underline{\gamma_s}}(t) \right\vert \qquad \text{since, } \vert h(n)\vert \leq \Vert h\Vert_{\infty} \\
        \leq \quad & \dfrac{2\Vert h\Vert_\infty}{\mu \left(\sum_{t = 1}^{\infty} M_{\overline{\gamma_s}}(t)\right)^2} \sum_{n=1}^\infty \sum_{t = 1}^\infty \left\vert M_{\overline{\gamma_s}}(n) M_{\underline{\gamma_s}}(t) \right\vert \\
        \leq \quad & \dfrac{2\Vert h\Vert_\infty}{\mu \left(\sum_{t = 1}^{\infty} M_{\overline{\gamma_s}}(t) \right)^2} \left( \sum_{n = 1}^\infty M_{\overline{\gamma_s}}(t) \right)^2 \quad \text{since, } M_{\underline{\gamma_s}}(t) \leq M_{\overline{\gamma_s}}(t) \text{ and both are nonnegative} \\
        = \quad & \dfrac{2}{\mu} \Vert h\Vert_\infty 
    \end{align*}
    
    Thus,
    \begin{equation*}
        \dfrac{1}{\Vert h\Vert_\infty}\left\vert \rho^{(s)}(\mu \bm{1} + h) - \rho^{(s)}(\mu \bm{1}) - \rho^{(s)'}(\mu \bm{1})(h) \right\vert < \dfrac{2B}{\mu} \Vert h\Vert_\infty
    \end{equation*}
    
    which goes to $0$ as $\Vert h\Vert_\infty \rightarrow 0$, where $\rho^{(s)'}(\mu \bm{1}) : l^\infty(\N) \rightarrow \R$ is;
    \begin{equation}
        \rho^{(s)'}(\mu \bm{1})(h) = \sum_{n = 1}^{\infty}\left[\dfrac{\sum_{t = 1}^{\infty} M_{\overline{\gamma_s}}(n) M_{\underline{\gamma_s}}(t) - \sum_{t = 1}^{\infty} M_{\underline{\gamma_s}}(n) M_{\overline{\gamma_s}}(t) }{\mu \left(\sum_{t = 1}^{\infty} M_{\overline{\gamma_s}}(t)\right)^2}\right] h(n)
        \label{eqnthm:frechet-diff-roughness}
    \end{equation}
    
    To see that the derivative given in eqn.~\ref{eqnthm:frechet-diff-roughness} is well defined, note that 
    
    $$\left\vert \rho^{(s)'}(\mu \bm{1})(h) \right\vert \leq \dfrac{2}{\mu} \Vert h \Vert_\infty,$$ 
    
    which shows that the infinite series is convergent as $h \in l^\infty(\N)$.
    
    This shows that the function $\rho^{(s)}$ given in eqn.~\ref{eqnthm:roughness-function} is Fr\'echet differentiable (see Miranda and Fichmann~\cite{de2005generalization} for definition) at $\mu \bm{1}$ and the derivative is given by eqn.~\ref{eqnthm:frechet-diff-roughness}. In a very similar way, the roughness measure corresponding to the completement set $\gamma_s^C$ can also be shown to be Fr\'echet differentiable at $\mu \bm{1}$ as a function of $\{ R(t): t = 0, 1, 2, \dots \}$. 
    
    Now, in order to obtain the exponential entropy as shown in eqn.~\ref{eqn:new-entropy}, we consider the function $g : [0, 1]^2 \rightarrow \R$ defined as;
    \begin{equation*}
        g(x, y) = xe^{(1 - x)} + y e^{(1 - y)}
    \end{equation*}
    
    Clearly, $g$ is Fr\'echet differentiable at every point of the domain. Finally, a chain rule~\cite{siddiqi2018functional,arora2019alternative} can be applied on the composition of $g$ and the roughness measures to show that, $\Psi_s$ is Fr\'echet differentiable at $\mu \bm{1}$ where $\Psi_s : l^\infty(\N) \rightarrow \R$ is such that $\Psi_s(R_{\delta_n}) = H^E_{\Delta, \delta_n, w}(s)$.
    
    Since Fr\'echet differentiablity implies Hadamard differentiability (see discussion followed by Definition 2.1 of Shao~\cite{shao1993}), and outputs a linear operator as the derivative, $\Psi_s$ is also Hadamard differentiable with the derivative given by;
    \begin{equation*}
        \Psi'_s\vert_{\mu \bm{1}(\cdot)}(h) = \sum_{n = 1}^{\infty}A_s(n)h(n)
    \end{equation*}
    
    where the function $A(n)$ is as given in eqn.~\ref{eqnthm:A-function}. Clearly, this is a linear and continuous map.
    
    The proof of the result now follows from an infinite dimensional generalization of delta method (Theorem 1 and Section 1.4 of R\"{o}misch~\cite{romisch2014delta}). Based on the discussion above, we see that due to the assumption, $Z$ is a gaussian process, and also the hadamard derivative of the function $\Psi_s$ is linear. Thus, $\Psi'_s(Z)$ is a normally distributed random variable with the mean $0$ and variance $\sigma^\ast$, as given in the statement of the theorem.
\end{proof}

An immediate extension of theorem~\ref{thm:delta-method-single} is the analogus result for multiple proposed changepoints.

\begin{corollary}
    \label{thm:delta-method-multiple}
    Under the same assumptions and conditions of theorem~\ref{thm:delta-method-single}, the vector of exponential entropy based statistic $H_{\Delta, \delta_n, w}^E(s)$ for multiple proposed changepoints $s_1, s_2, \dots s_k$ has the following asymptotic distribution as $a_{\delta_n} \rightarrow \infty$;
    
    \begin{equation*}
        a_{\delta_n} 
        \begin{bmatrix}
           H_{\Delta, \delta_n, w}^E(s_1) - H^\ast(s_1)\\
           H_{\Delta, \delta_n, w}^E(s_2) - H^\ast(s_2)\\
           \dots \\
           H_{\Delta, \delta_n, w}^E(s_k) - H^\ast(s_k)\\
        \end{bmatrix} \rightsquigarrow
        \norm_k(\bm{0}_k, \Sigma^\ast_{k\times k})
    \end{equation*}
    
    where $\bm{0}_k$ is the $k$ dimensional null vector, and the entries of the $k\times k$ dispersion matrix are given as;
    \begin{equation*}
        \left( \Sigma^\ast_{k \times k} \right)_{(i, j)} = \sum_{m=1}^\infty \sum_{n=1}^\infty A_{s_i}(m) \sigma(m, n) A_{s_j}(n) \qquad i, j = 1, 2, \dots k
    \end{equation*}
    
    In particular, if $\vert s_i - s_j \vert > (4w + 2\Delta)$, then $\left( \Sigma^\ast_{k \times k} \right)_{(i, j)} = 0$.
\end{corollary}

\begin{algorithm}[ht]
\SetAlgoLined
\SetKwInOut{inputfield}{Input}
\SetKwInOut{outputfield}{Output}
\inputfield{ A multivariate time series $Y_t$ for $t = 1, 2, \dots T$ with $Y_t \in \R^p$, Parameters $\delta, \Delta$ and $w$, An acceptance threshold $\alpha$ (usally $0.05$)}
\outputfield{ Estimated changepoints $s_1, s_2, \dots s_k$}
\For{$t = 1$ to $T$}{
    Create two samples $\mathcal{S}_{1t} = \{ Y_{\max\{1, t-\delta + 1\} }, \dots, Y_{t} \}$ and $\mathcal{S}_{2t} = \{ Y_{t+1}, \dots, Y_{ \min\{ T, t+\delta \} } \}$\;
    Compute regularity measure $R(t)$ as a similarity between $\mathcal{S}_{1t}$ and $\mathcal{S}_{2t}$\;
}
\For{$s = 1$ to $T$}{
    \For{$t = 1$ to $T$}{
        Use theorem~\ref{thm:exact-upper-lower} to compute $M_{\underline{\gamma_s}}(t)$\;
        Use lemma~\ref{thm:approx-symmetric} to compute $M_{\overline{\gamma_s}}(t)$\;
        Use lemma~\ref{thm:approx-complement} to compute $M_{\underline{\gamma_s^C}}(t)$ and $M_{\overline{\gamma_s^C}}(t)$\;
    }
    Compute roughness measure $\rho_{\Delta, \delta, w}(\gamma_s)$ and $\rho_{\Delta, \delta, w}(\gamma_s^C)$ using eqn.~\ref{eqn:new-roughness}\;
    Compute exponential entropy $H^E_{\Delta, \delta, w}(s)$ using eqn.~\ref{eqn:new-entropy}\;
}
Find local minimas of the sequence $\{ H^E_{\Delta, \delta, w}(s) : s = 1, 2, \dots T \}$\;
\hrulefill\\
\textit{One can output these local minimas as proposed changepoints, however if theoretical expressions for $\mu$ and $\sigma$ as denoted in theorem~\ref{thm:delta-method-single} for the underlying regularity measure $R(t)$ is available, the following testing framework can be used to further enhance its ability}\\
\hrulefill\\
\For{each local minima $s_k$ of $H^E_{\Delta, \delta, w}$}{
    Compute $H^\ast(s_k)$ as given in theorem~\ref{thm:delta-method-single} by replacing the infinite sums as sums from $t = 1$ to $T$\;
    \For{$n = 1$ to $T$}{
        Use eqn.~\ref{eqnthm:A-function} to compute $A_{s_k}(n)$ by replacing the infinite sums as sums from $t = 1$ to $T$\;
    }
    Compute $\sigma^\ast$ for the changepoint $s_k$ as given in theorem~\ref{thm:delta-method-single} using finite sums from $t = 1$ to $T$\;
    Accept $s_k$ as a changepoint if $a_{\delta}\left( H^E_{\Delta, \delta, w}(s_k) - H^\ast(s_k) \right)/\sqrt{\sigma^\ast} < z_{\alpha}$, where $z_\alpha$ is the $\alpha$-th quantile of the standard normal distribution\;
}
\caption{Rough-Fuzzy CPD}
\label{algo:main}
\end{algorithm}

On the basis of theorem~\ref{thm:delta-method-single} and theorem~\ref{thm:delta-method-multiple}, a hypothesis testing framework can be laid out to detect the significant changepoints. The final algorithm is shown in algorithm~\ref{algo:main}. Another consequence of theorem~\ref{thm:delta-method-multiple} is that if the parameters $w$ and $\Delta$ are carefully chosen so that the obtained local minima of the entropy function $H^E_{\Delta, \delta, n}(s)$ are separated by atleast $(4w + 2\Delta)$, then the entropy corresponding to these functions are asymptotically independent. Thus, for such well-separated changepoints, each of them can be tested individually for false positives.

\section{Simulation Studies}\label{sec:simulation}

\subsection{Simulation Setups}\label{sec:simulation-scenario}

In order to asses a comparative study of the proposed rough-fuzzy improvement over any changepoint detection algorithm, some simulations are performed. We consider a general model of the observations as follows:

\begin{equation*}
    y_t = \mu(t) + \epsilon_t \qquad \epsilon_t \sim \norm(0, 1)
\end{equation*}

where $\mu(t)$ is the mean function dependent on time, while $\epsilon_t$ denotes independent and identically distributed random variables, each distributed according to a standard normal distribution. We consider $3$ different types of mean curves as shown in figure~\ref{fig:cp-mean-curves}.

\begin{figure}[ht]
    \centering
    \begin{subfigure}[b]{0.33\linewidth}
        \includegraphics[width = \textwidth]{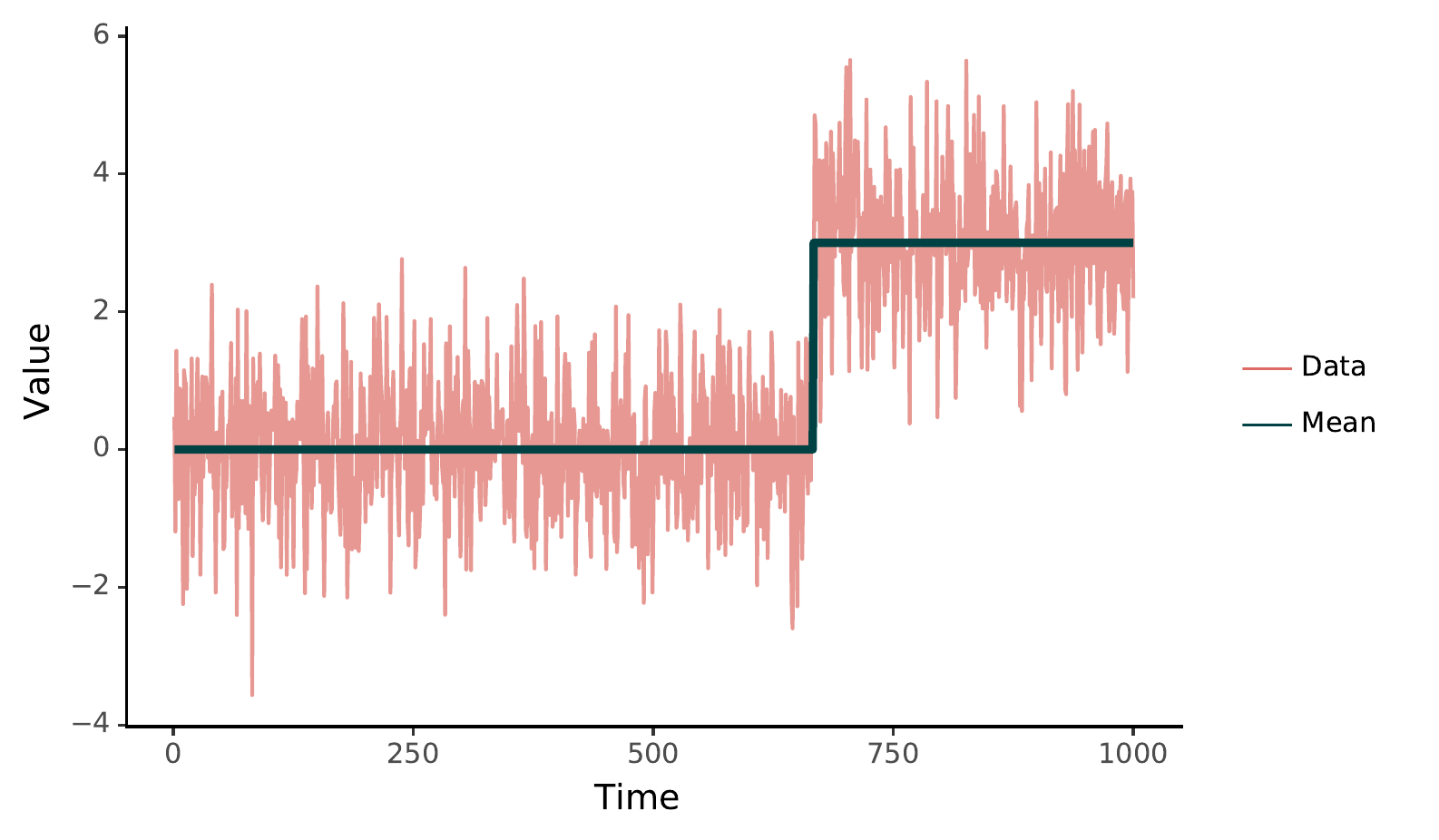}
        \caption{Discrete Jump in mean function ($S1$)}
        \label{subfig:discrete-jump}
    \end{subfigure}
    \hfill
    \begin{subfigure}[b]{0.33\linewidth}
        \includegraphics[width = \textwidth]{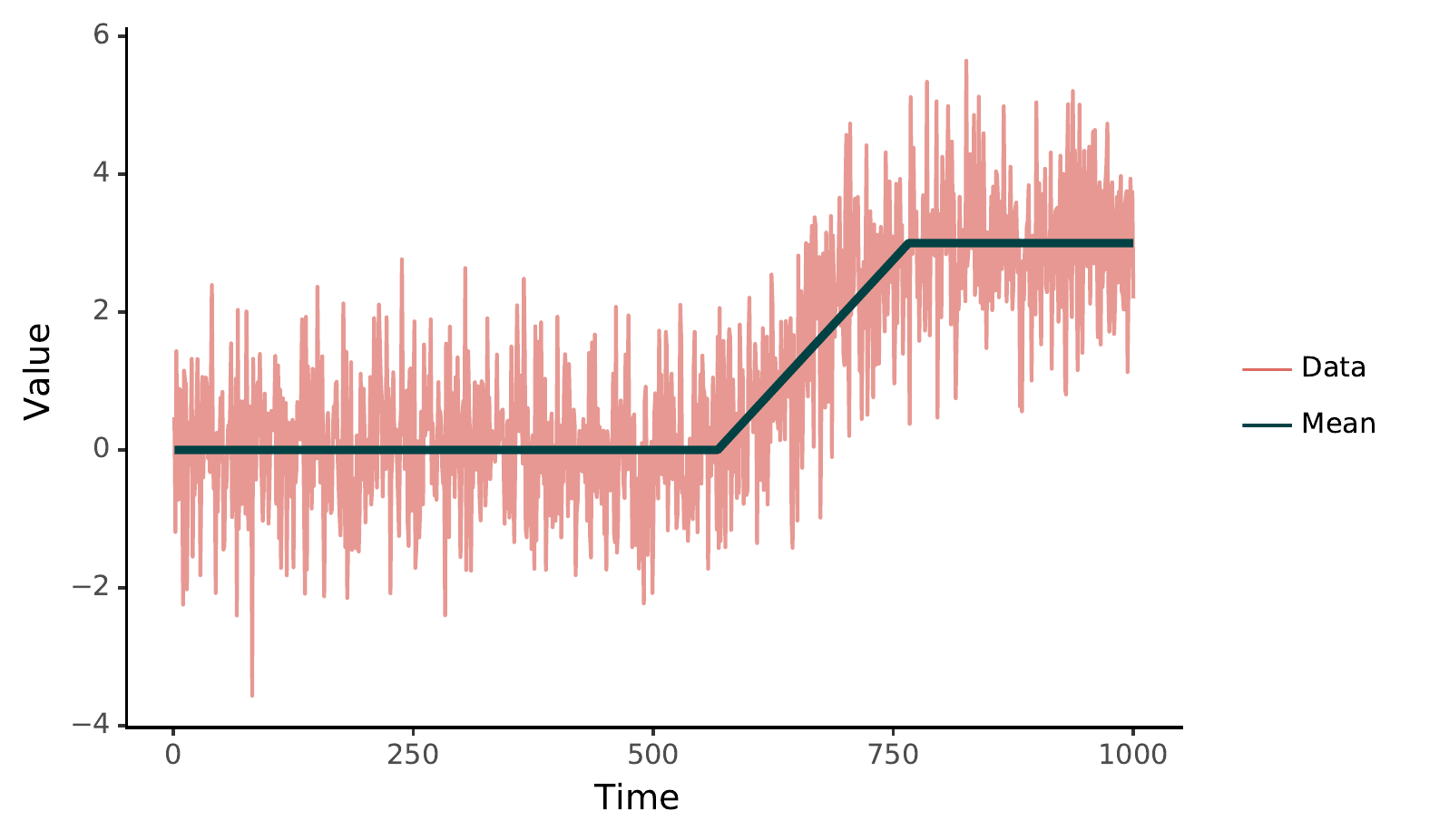}
        \caption{Continuous Jump in mean function ($S2$)}
        \label{subfig:cont-jump}
    \end{subfigure}
    \hfill
    \begin{subfigure}[b]{0.33\linewidth}
        \includegraphics[width = \textwidth]{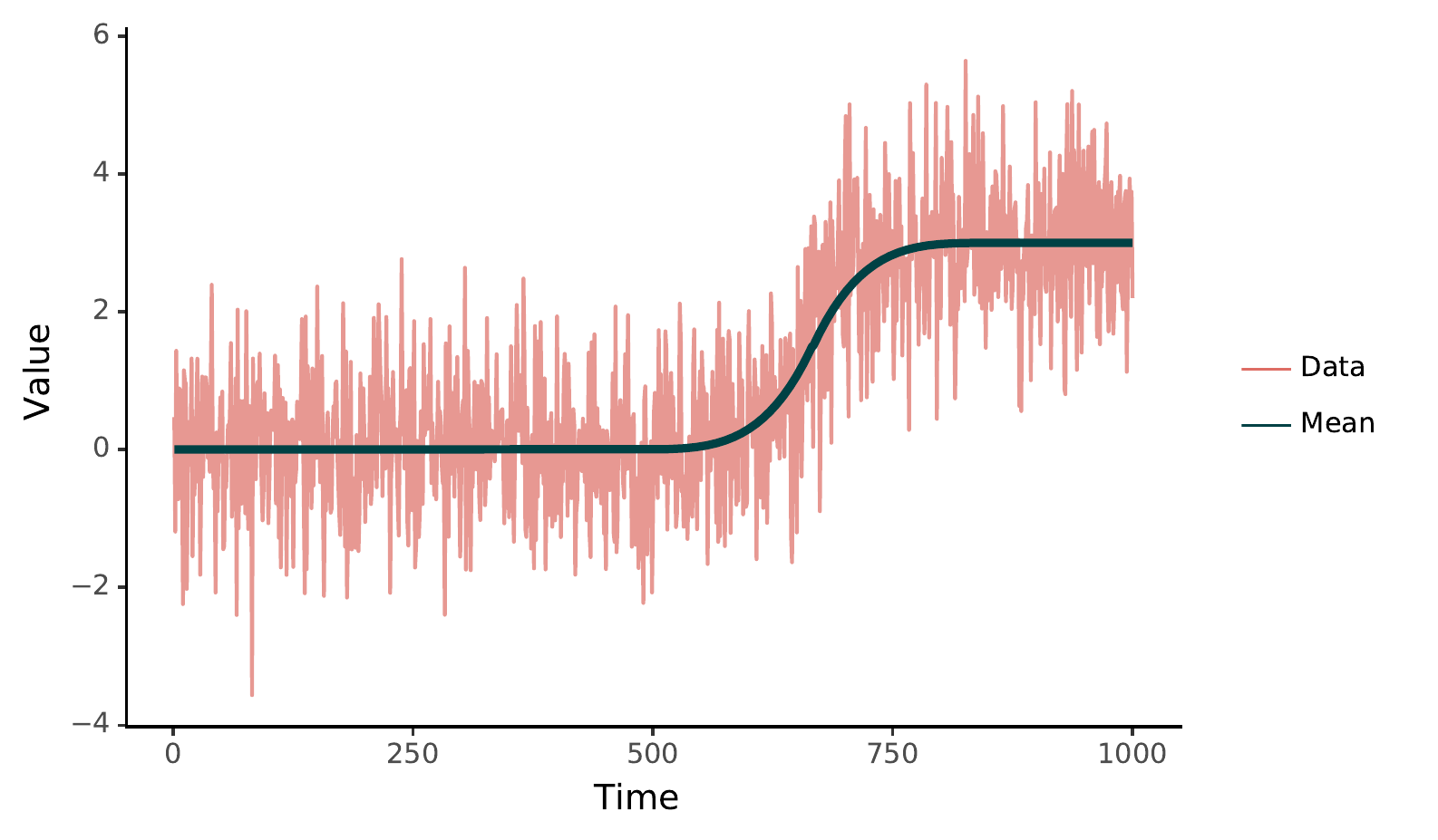}
        \caption{Smooth Jump in mean function ($S3$)}
        \label{subfig:fuzzy-jump}
    \end{subfigure}
    \caption{Different Scenarios of changepoints with Changes in mean of the data}
    \label{fig:cp-mean-curves}
\end{figure}

Figure \ref{subfig:discrete-jump} depicts the scenario $S1$ of abrupt change, which is the popular definition of a changepoint. Most of the existing algorithms are able to detect this type of change. Figure \ref{subfig:cont-jump} shows the scenario $S2$ where the mean curve is piecewise linear, and hence a test for changes in slope of the regression would be able to detect such changes. While this shows a situation where the change in slope of the regression equation has a discrete jump, situation $S3$ indicates a continual change in the slope thus incorporating a much smoother and gradual jump in the dataset as shown in subfigure~\ref{subfig:fuzzy-jump}. 

In order to establish the comparison, we consider three test statistic that are well used in a two-sample setup to detect the changes in mean. The parametric statistic as square of two sample t-test statistic, nonparametric Kolmogorov Smirnov test statistic and Augmented Dickey Fuller test statistic among unit root tests were chosen. In each case, the regularity measure $R(t)$ is obtained by taking reciprocal of the test statistics, to make sure $R(t)$ has a higher value in absence of changepoints and has a lower value in presence of changepoint. Both eqn.~\ref{eqn:old-cp} and eqn.~\ref{eqn:new-cp} were used to obtain the estimates of the changepoints, and the process was repeated for $200$ Monte Carlo resamples and used to obtain error measures. In applying rough-fuzzy CPD, $\delta = 50$ was chosen to compute the regularity measure, while different values of $w$ and $\Delta$ were used in the experiment to see their effects. 

We present the results of simulation studies for all the $3$ cases described above. All the following simulations and real data applications have been performed by the python package \hyperlink{https://pypi.org/project/roufcp/}{\texttt{roufcp}} developed by us.

\subsection{Result of Simulation}

\begin{table}[ht]
\centering
\begin{tabular}{@{}llllll@{}}
\toprule
\multicolumn{1}{c}{\multirow{2}{*}{\textbf{\begin{tabular}[c]{@{}c@{}}Simulation \\ setup\end{tabular}}}} &
  \multicolumn{1}{c}{\multirow{2}{*}{\textbf{\begin{tabular}[c]{@{}c@{}}Base test\\ statistic\end{tabular}}}} &
  \multicolumn{1}{c}{\multirow{2}{*}{\textbf{\begin{tabular}[c]{@{}c@{}}Parameters for\\ proposed method\end{tabular}}}} &
  \multicolumn{2}{c}{\textbf{RMSE}} &
  \multicolumn{1}{c}{\multirow{2}{*}{\textbf{\begin{tabular}[c]{@{}c@{}}Relative\\ Decrease in MSE\end{tabular}}}} \\ \cmidrule(lr){4-5}
\multicolumn{1}{c}{} &
  \multicolumn{1}{c}{} &
  \multicolumn{1}{c}{} &
  \multicolumn{1}{c}{\textbf{Proposed}} &
  \multicolumn{1}{c}{\textbf{Base}} &
  \multicolumn{1}{c}{} \\ \midrule
\multirow{12}{*}{\begin{tabular}[c]{@{}l@{}}Discrete\\ Jump\\ in Mean\\ (S1)\end{tabular}} &
  \multirow{4}{*}{t-test} &
  $w = \Delta = 5$ &
  1.591 &
  \multirow{4}{*}{1.752} &
  17.53\% \\
 &
   &
  $w = \Delta = 25$ &
  4.961 &
   &
  -701.81\% \\
 &
   &
  $w = \Delta = 50$ &
  6.567 &
   &
  -1436.26\% \\
 &
   &
  $w = \Delta = 100$ &
  21.446 &
   &
  -14883.88\% \\ \cmidrule(l){2-6} 
 &
  \multirow{4}{*}{KS-test} &
  $w = \Delta = 5$ &
  2.001 &
  \multirow{4}{*}{3.928} &
  74.08\% \\
 &
   &
  $w = \Delta = 25$ &
  4.926 &
   &
  -57.27\% \\
 &
   &
  $w = \Delta = 50$ &
  7.404 &
   &
  -255.29\% \\
 &
   &
  $w = \Delta = 100$ &
  25.533 &
   &
  -4125.33\% \\ \cmidrule(l){2-6} 
 &
  \multirow{4}{*}{ADF-test} &
  $w = \Delta = 5$ &
  38.115 &
  \multirow{4}{*}{106.684} &
  87.23\% \\
 &
   &
  $w = \Delta = 25$ &
  41.395 &
   &
  84.94\% \\
 &
   &
  $w = \Delta = 50$ &
  43.843 &
   &
  83.11\% \\
 &
   &
  $w = \Delta = 100$ &
  55.305 &
   &
  73.12\% \\ \midrule
\multirow{9}{*}{\begin{tabular}[c]{@{}l@{}}Continuous\\ Jump in \\ Mean\\ (S2)\end{tabular}} &
  \multirow{3}{*}{t-test} &
  $w = \Delta = 25$ &
  25.232 &
  \multirow{3}{*}{35.501} &
  49.48\% \\
 &
   &
  $w = \Delta = 50$ &
  11.326 &
   &
  89.82\% \\
 &
   &
  $w = \Delta = 100$ &
  17.419 &
   &
  75.92\% \\ \cmidrule(l){2-6} 
 &
  \multirow{3}{*}{KS-test} &
  $w = \Delta = 25$ &
  26.881 &
  \multirow{3}{*}{79.723} &
  88.63\% \\
 &
   &
  $w = \Delta = 50$ &
  15.071 &
   &
  96.42\% \\
 &
   &
  $w = \Delta = 100$ &
  23.778 &
   &
  91.10\% \\ \cmidrule(l){2-6} 
 &
  \multirow{3}{*}{ADF-test} &
  $w = \Delta = 25$ &
  207.811 &
  \multirow{3}{*}{347.614} &
  64.26\% \\
 &
   &
  $w = \Delta = 50$ &
  197.768 &
   &
  67.63\% \\
 &
   &
  $w = \Delta = 100$ &
  189.827 &
   &
  70.18\% \\ \midrule
\multirow{9}{*}{\begin{tabular}[c]{@{}l@{}}Smooth\\ Jump in\\ Mean\\ (S2)\end{tabular}} &
  \multirow{3}{*}{t-test} &
  $w = \Delta = 25$ &
  12.242 &
  \multirow{3}{*}{18.001} &
  53.75\% \\
 &
   &
  $w = \Delta = 50$ &
  9.596 &
   &
  71.58\% \\
 &
   &
  $w = \Delta = 100$ &
  18.768 &
   &
  -8.70\% \\ \cmidrule(l){2-6} 
 &
  \multirow{3}{*}{KS-test} &
  $w = \Delta = 25$ &
  11.640 &
  \multirow{3}{*}{37.714} &
  90.47\% \\
 &
   &
  $w = \Delta = 50$ &
  10.127 &
   &
  92.78\% \\
 &
   &
  $w = \Delta = 100$ &
  22.936 &
   &
  63.01\% \\ \cmidrule(l){2-6} 
 &
  \multirow{3}{*}{ADF-test} &
  $w = \Delta = 25$ &
  197.282 &
  \multirow{3}{*}{286.67} &
  52.64\% \\
 &
   &
  $w = \Delta = 50$ &
  187.246 &
   &
  57.33\% \\
 &
   &
  $w = \Delta = 100$ &
  181.6 &
   &
  59.87\% \\ \bottomrule
\end{tabular}\\
\vspace{1em}
\caption{Comparison of rough-fuzzy CPD (proposed) vs usual (base) methods based on cost minimization for 3 types of mean shift changepoints}
\label{tbl:sim-comparison}
\end{table}

Table \ref{tbl:sim-comparison} summarizes the result for all 3 types of mean shift changepoints described above. We begin analyzing the results with changepoint of type $S1$, where a discrete jump in the mean function has occurred. Note that, the hyperparameters $w$ and $\Delta$ control the amount of fuzziness and roughness to incorporate when detecting changepoints. Clearly, $w = \Delta = 0$ would entail just a transformation of $R(t)$ as the exponential entropy, and both the estimating equations eqn.~\ref{eqn:old-cp} and eqn.~\ref{eqn:new-cp} would yield the same changepoint. Thus, in general, increasing $w$ and $\Delta$ would increase RMSE (root mean square error) if the true model has a discrete jump change in mean function. Indeed, for $w = \Delta = 5$, for KS test, we have a $74\%$ decrease in MSE while with higher values of $w$ and $\Delta$, we see higher RMSE for rough-fuzzy CPD compared to the base methods using t-test and KS-test. For ADF test statistic, however, the proposed method outperforms the base method across all hyperparameter values, achieving more than $70\%$ reduction in MSE. Hence, we see it is important to use appropriate values of hyperparameters. Overall, for discrete jump type changepoints rough-fuzzy CPD performs poorly, giving higher MSE than base model for t-test and Kolmogorov Smirnov tests, and higher values of $w$ and $\Delta$ will decrease the accuracy in prediction. This is because the fundamental assumption of the proposed method viz. "change is fuzzy in nature and not abrupt" is violated in this case. Now, we look at the other 2 types of changepoints where the change is not abrupt and occurs gradually over a period of time.\par

From Table \ref{tbl:sim-comparison} we observe the performance of rough-fuzzy CPD for situation $S2$ with continuous change in mean function. For both t-test statistic and Kolmogorov-Smirnov statistic, the RMSE is much lower in rough-fuzzy CPD relative to the base method based on the regularity measure $R(t)$. However, with an increase in $w$ and $\Delta$, the RMSE reduces first and then increases, possibly suggesting an existence of optimal hyperparameters in between. Because of the specific parametric setup, Augmented Dickey Fuller test and nonparametric Kolmogorov Smirnov test generally perform worse than parametric t-test. However, in practical applications when the data generating processes are not known, ADF and KS test might perform better in conjunction with our proposed improvement.

\begin{figure}[ht]
    \centering
    \includegraphics[width = \linewidth]{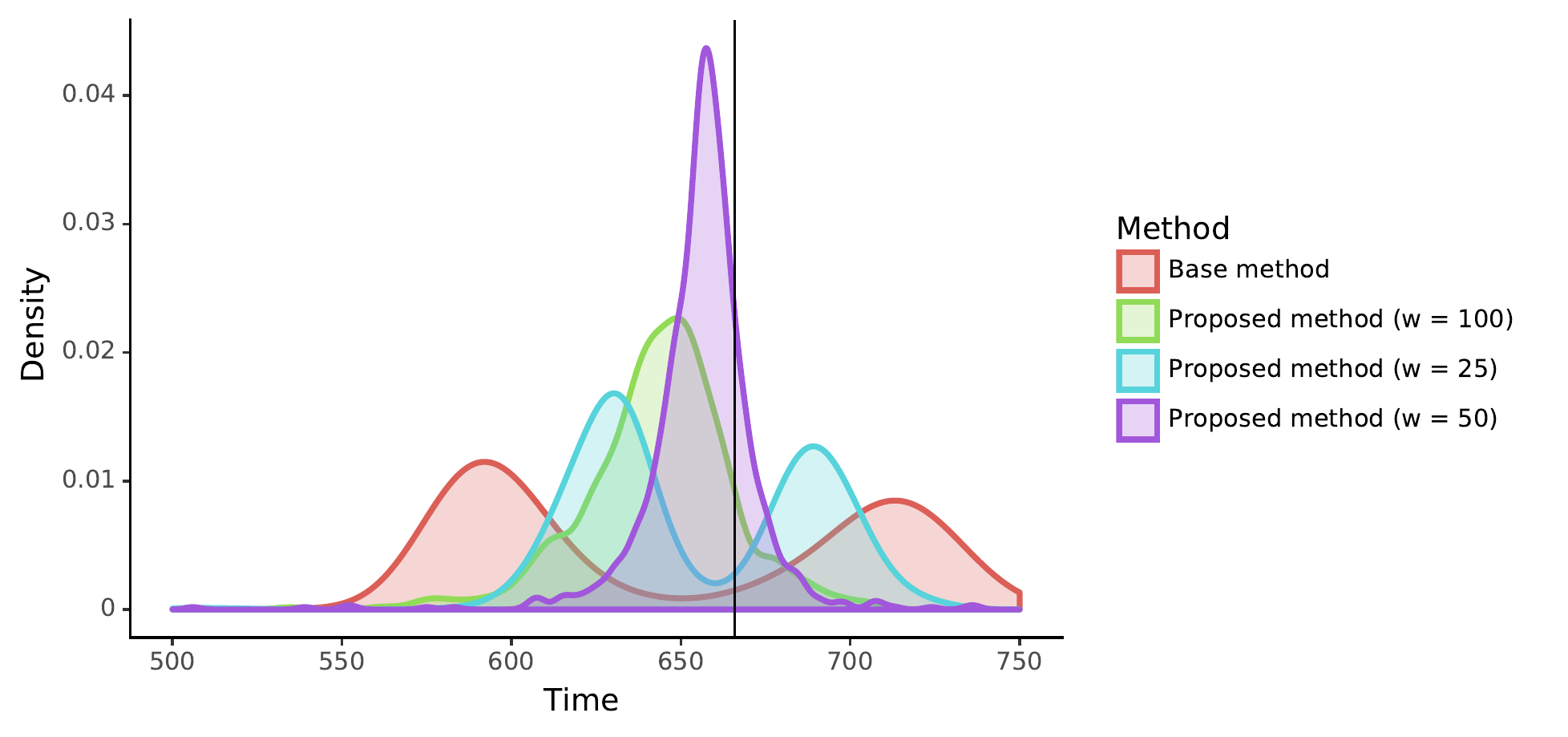}
    \caption{Base method along with proposed rough-fuzzy CPD based on the regularity measure with Augmented Dickey Fuller statistic for continuous jump changepoint ($S2$)}
    \label{fig:adf-all}
\end{figure}

An interesting phenomenon occurs when we use Augmented Dickey Fuller test statistic based regularity measure. As shown in figure~\ref{fig:adf-all}, the ADF test usually outputs two changepoints around the true changepoint at $666$; but as the fuzziness is incorporated in the model by increasing $w$ and $\Delta$, the bimodal distribution gradually becomes unimodal. However, there possibly remains a negative bias in estimation, as indicated by the mode of the density curve of the rough-fuzzy CPD for $w = \Delta = 50$. \par 

Turning to situation $S3$ with smooth change in mean function, the results are found to be similar to scenario $S2$. Here also rough-fuzzy CPD has reduced the MSE by more than $50\%$ in most cases. However, we see that for $w = \Delta = 100$ the efficiency of model is greatly reduced and it even performs slightly worse than base t-test. Again, as mentioned earlier, it is of utmost importance to choose the values of the hyperparameters $w$ and $\Delta$ carefully to obtain accurate estimates of the changepoint even in the situations where the mean function is gradual and the underlying assumption of the method is not violated. However, as shown by table \ref{tbl:sim-comparison}, rough-fuzzy CPD generally obtains higher reduction in MSE for $S2$ (continuous jump) and $S3$ (smooth jump) compared to $S1$ (discrete jump). So, with greater degree of ambiguity in changepoint, the relative performance of our model increases as expected. We shall illustrate this phenomenon in greater details and more rigorously in subsection~\ref{subsec:effect of fuzzyness}.

\subsection{Performance under different signal-to-noise ratios}\label{subsec:effect of snr}
While detecting fuzzy changepoints is a hard problem, doing so in a noisy data with low Signal-to-Noise Ratio (SNR) is even harder. Here signal-to-noise ratio is defined as 

\[ SNR = \dfrac{\mathbb{E}(S^2)}{\mathbb{E}(N^2)} \]

where $N$ is the noise component of the data and $S$ is the true signal component of the data, which mainly comprises of the mean function in the time series observations. To check the performance of rough-fuzzy CPD under different SNR's, we keep noise $N$ same as before with standard normally distributed and vary the value of signal $S$ by changing the size of the jump in the mean function. In this simulation setup, we consider scenario $S2$ with continuous change in mean function, with the jump sizes $S = = 1/5, 1/4, 1/3, 1/2, 1,2,3, 4,5, 6, 7,8, 9,10$. For each such setup, monte carlo estimates of the MSE of the estimated changepoints by the proposed method rough-fuzzy CPD and the base method using Kolmogorov-Smirnov test, are calculated based on $1000$ resamples. In each resample, the true changepoint value is kept fixed at $666$. 

\begin{figure}[ht]
    \centering
    \includegraphics[width = \linewidth]{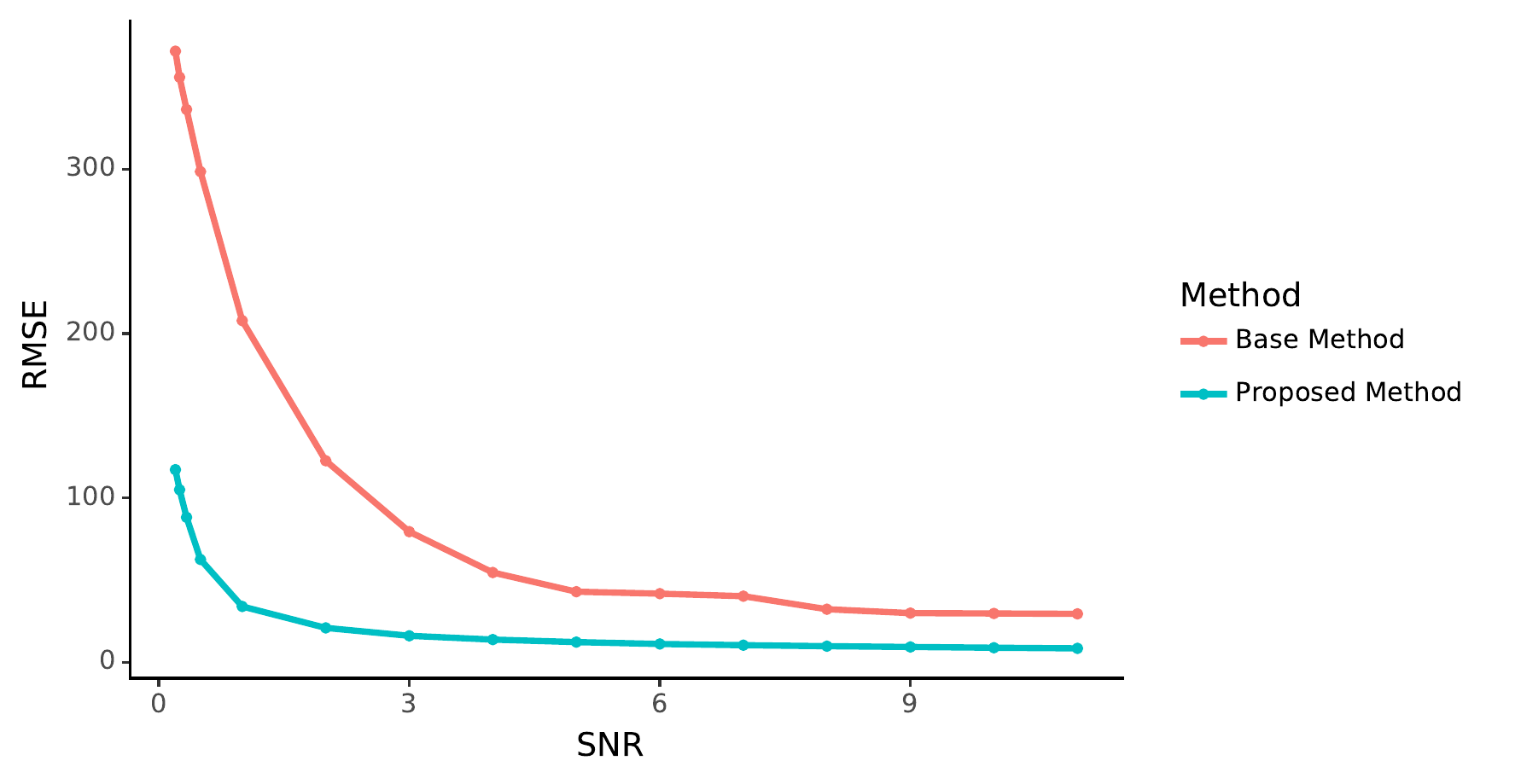}
    \caption{Performance of proposed rough-fuzzy CPD and base method (KS-test)for different values of signal-to-noise ratio}
    \label{fig:effect_of_snr}
\end{figure}

Figure \ref{fig:effect_of_snr} shows the variation of RMSE of estimated changepoint for the proposed rough-fuzzy CPD as well as the base method. We observe that rough-fuzzy CPD has obtained lower RMSE than the base method for all values of SNR. While SNR increases, the predictive capabilities of both the methods increase. However, it is important to note that while the base method performs very badly for lower values of SNR ($ < 1$), the performance of rough-fuzzy CPD using the same regularity measure is not so severely affected. In fact throughout the low and high values of SNR, the relative improvement of MSE achieved by rough-fuzzy CPD remains fairly uniform ranging from $88-97\%$, with the highest reduction achieved when SNR is 2 or 3.

\subsection{Effect of fuzziness of true changepoint on efficiency}\label{subsec:effect of fuzzyness}
The main aim of this paper is to provide a method of changepoint detection when the change is gradual and not abrupt. In our earlier simulations, scenario $S2$ and $S3$ depict such situations where such a gradual change has occurred. To check how rough-fuzzy CPD performs under different levels of graduality or "fuzziness" of change, we consider different cases of scenario $S2$ where the mean function is defined as follows;
\begin{align*}
    \mu_{s,\Delta}(t) = 
        \begin{cases}
          0 &  t \leq s_0- \mathcal{F}\\
          \dfrac{t-(s_0 -\mathcal{F})}{\mathcal{F}} & s_0 -\mathcal{F} < t \leq s_0+\mathcal{F} \\
          2 &  t > s_0+\mathcal{F}
        \end{cases}
\end{align*}
where $s_0$ is the true changepoint and the parameter $\mathcal{F}$ defines the degree of graduality or "fuzziness" of the changepoint. To see the effect of $\mathcal{F}$ on the estimate of rough-fuzzy CPD, we consider $1000$ resamples of scenario $S2$ with true changepoint at $s_0 = 666$, for $15$ different values of $\mathcal{F}$ ranging from $10$ to $150$. The MSE for each level of fuzziness $\mathcal{F}$ is calculated using a monte-carlo method, for the estimated changepoint obtained by rough-fuzzy CPD as well as the underlying base method with Kolmogorov-Smirnov statistic.

\begin{figure}[ht]
    \centering
        \begin{subfigure}[b]{0.59\linewidth}
            \includegraphics[width = \textwidth]{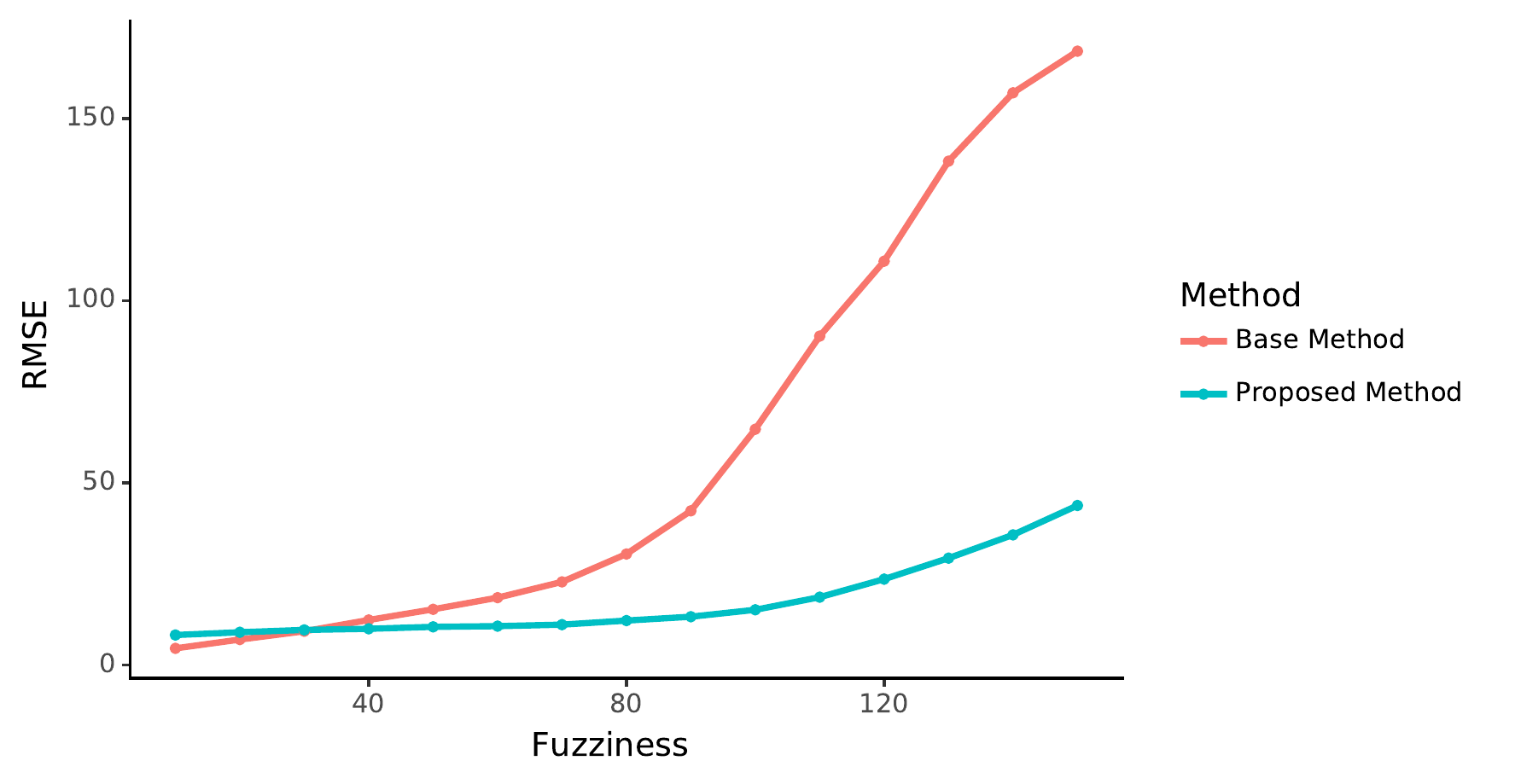}
            \caption{RMSE of base model (KS) and proposed model}
            \label{subfig:effect_of_fuzziness_rmse}
        \end{subfigure}
        \hfill
        \begin{subfigure}[b]{0.39\linewidth}
            \includegraphics[width = \textwidth]{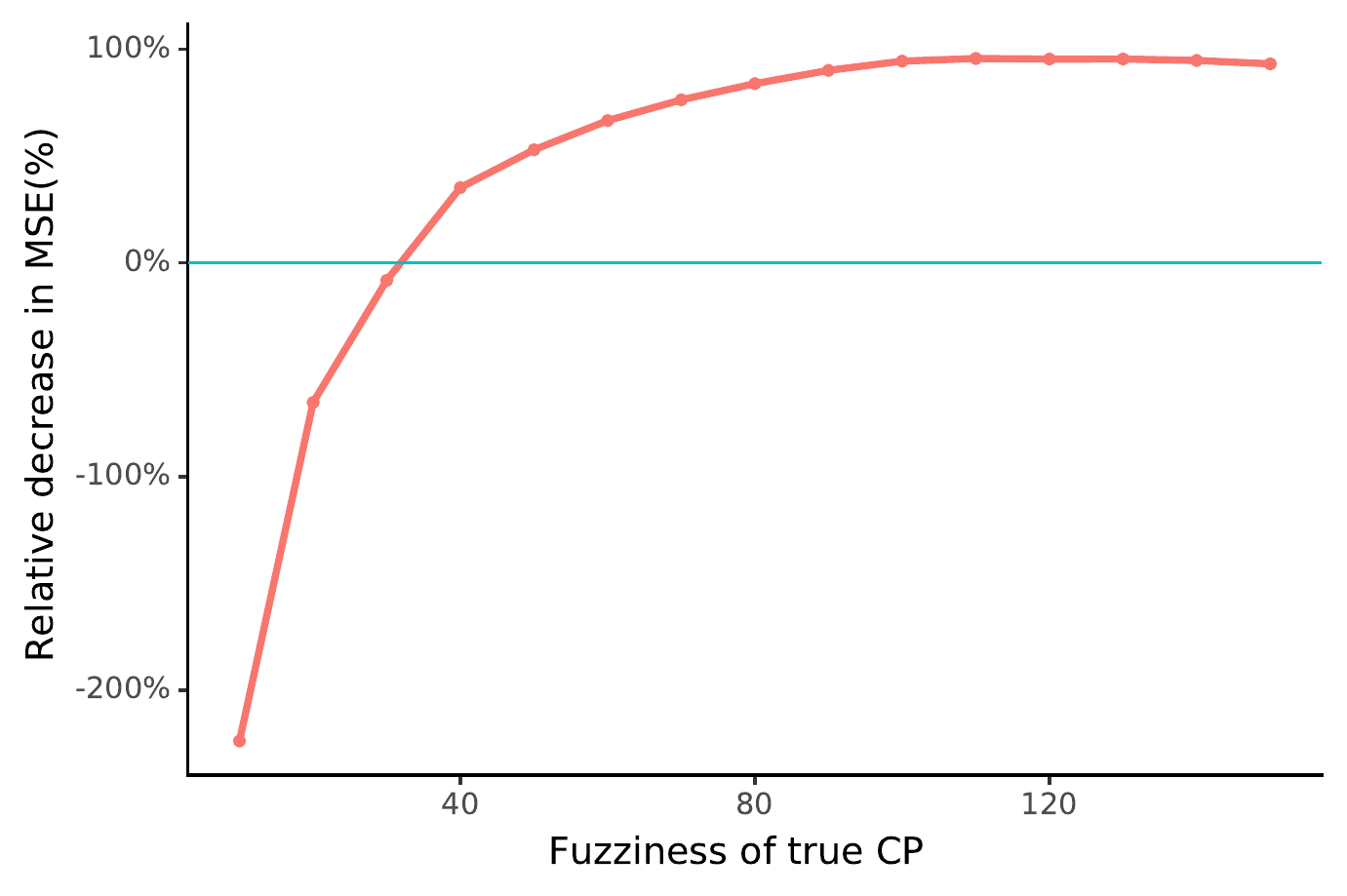}
            \caption{Relative decrease in MSE compared to base model}
            \label{subfig:effect_of_fuzziness_relative_mse}
    \end{subfigure}
    \caption{Performance of proposed rough-fuzzy CPD and base method (KS-test)for different values of fuzziness of true changepoint}
    \label{fig:effect_of_fuzziness}
\end{figure}

Figure \ref{fig:effect_of_fuzziness} shows the RMSE of estimated changepoints by rough-fuzzy CPD and its base counterpart for different values of fuzziness $\mathcal{F}$, along with the relative decrease obtained by rough-fuzzy CPD. As shown in subfigure~\ref{subfig:effect_of_fuzziness_rmse}, higher degree of fuzziness leads to higher error in estimation for both methods, though rough-fuzzy CPD is not as severely affected as the base method. Figure~\ref{subfig:effect_of_fuzziness_relative_mse} depicts that the relative performance gain by the proposed rough-fuzzy improvement is only possible when the fuzziness in the true changepoint crosses a certain threshold, about $\mathcal{F} = 31$ for our specific setup. However, as the fuzziness increases, due to extremely sensitive performance of the base method, rough-fuzzy CPD could achieves nearly $95-99\%$ performance gain in terms of MSE.

\subsection{Comparison with other methods}

To compare performance of the rough-fuzzy CPD with existing changepoint detection algorithms, we have chosen three existing methods of fuzzy changepoint detection, namely FCP algorithm~\cite{chang2015fuzzy} for regression models, Fuzzy shift changepoint (FSCP) algorithm~\cite{lu2016change} and Fuzzy classification maximum likelihood changepoint (FCMLCP) algorithm~\cite{lu2016detecting}. Complementary to that, we also consider $3$ state of the art abrupt changepoint detection algorithms, namely Wild Binary Segmentation~\cite{fryzlewicz2014wild}, Pruned Exact Linear Time (PELT)~\cite{killick2012optimal} and Bayesian Online changepoint Detection~\cite{adams2007bayesian,knoblauch2018spatio}. To compare the performances between rough-fuzzy CPD and these existing methods, we consider rough-fuzzy CPD with "best" tuning parameter in combination with Kolmogorov-Smirnov statistic, which should perform equally well under any general model due to its nonparametric nature.

\begin{table}[ht]
\centering
\begin{tabular}{@{}lll@{}}
\toprule
\textbf{Simulation setup} & \textbf{Algorithm} & \textbf{RMSE} \\ \midrule
\multirow{7}{*}{Discrete jump in Mean (S1)} & FSCP & 2.107 \\
 & FCP & 1.649 \\
 & FCMLCP & 5.53 \\
 & Wild Binary Segmentation* & 1.857 \\
 & PELT & 1.857 \\
 & BOCD & 6.23 \\
 & Rough-Fuzzy CPD (using KS statistic) & 2.001 \\ \midrule
\multirow{7}{*}{Continuous jump in Mean (S2)} & FSCP & 28.87 \\
 & FCP & 190.11 \\
 & FCMLCP & 52.39 \\
 & Wild Binary Segmentation & 18.028 \\
 & PELT & 29.036 \\
 & BOCD & 15.478 \\
 & Rough-Fuzzy CPD (using KS statistic)* & 15.071 \\ \midrule
\multirow{7}{*}{Smooth jump in Mean (S3)} & FSCP & 17.436 \\
 & FCP & 117.35 \\
 & FCMLCP & 37.29 \\
 & Wild Binary Segmentation & 14.265 \\
 & PELT & 21.424 \\
 & BOCD & 11.086 \\
 & Rough-Fuzzy CPD (using KS statistic)* & 10.127 \\ \bottomrule
\end{tabular}\\
\vspace{1em}
\caption{Performance comparison of rough-fuzzy CPD (proposed) with existing fuzzy changepoint detection methods under the simulation setups. Best method in each setup is indicated by $\ast$.}
\label{tab:compare-methods}
\end{table}

Table \ref{tab:compare-methods} shows the root mean squared error (RMSE) of the estimated changepoints by the proposed algorithm rough-fuzzy CPD and the other existing changepoint detection methods, for the simulation scenarios described in section~\ref{sec:simulation-scenario}, with the best model indicated for each scenario. As shown in table~\ref{tab:compare-methods}, in the scenario $S2$ and $S3$ where fuzziness about the changepoint is present, the proposed model outperforms the existing methods. However, in case of abrupt changepoints, PELT, Wild Binary Segmentation and FCP are more suited to the task. The Bayesian Online changepoint detection method turns out to be very close to our method in terms of performance for smooth change in mean, but at the cost of inaccurate estimate for abrupt changes. One possible reason for this phenomenon could be the underlying nature of BOCD, which being an online changepoint detection algorithm, determines presence of a changepoint only from the past data. In such situation, once a true changepoint is visited, BOCD needs to see sufficient samples from the new distributions before it can conclude such a changepoint is present, thus creating a lag between the true changepoint and the estimated changepoint. This inability in abrupt changepoint becomes beneficial in detecting gradual changepoint for BOCD.

\subsection{Sensitivity Analysis}\label{sec:sensitivity analysis}
The proposed method rough-fuzzy CPD depends on 3 main factors viz. the regularity measure $R(t)$, the tunable hyperparameters $w$ and $\Delta$. $w$ denotes the degree of roughness of the tolerance function with higher values indicating greater roughness. Similarly, the parameter $\Delta$ determines the fuzziness of the membership function with higher values corresponding to greater fuzziness. While we do not have control over the regularity measure $R(t)$ as it is exogenous to rough-fuzzy CPD, we need to choose the values of $w$ and $\Delta$ judiciously to obtain better estimates. To understand the effect of parameters $w$ and $\Delta$ we perform some sensitivity analysis.
\begin{figure}[ht]
    \centering
    \begin{subfigure}[b]{0.49\linewidth}
        \includegraphics[width = \textwidth]{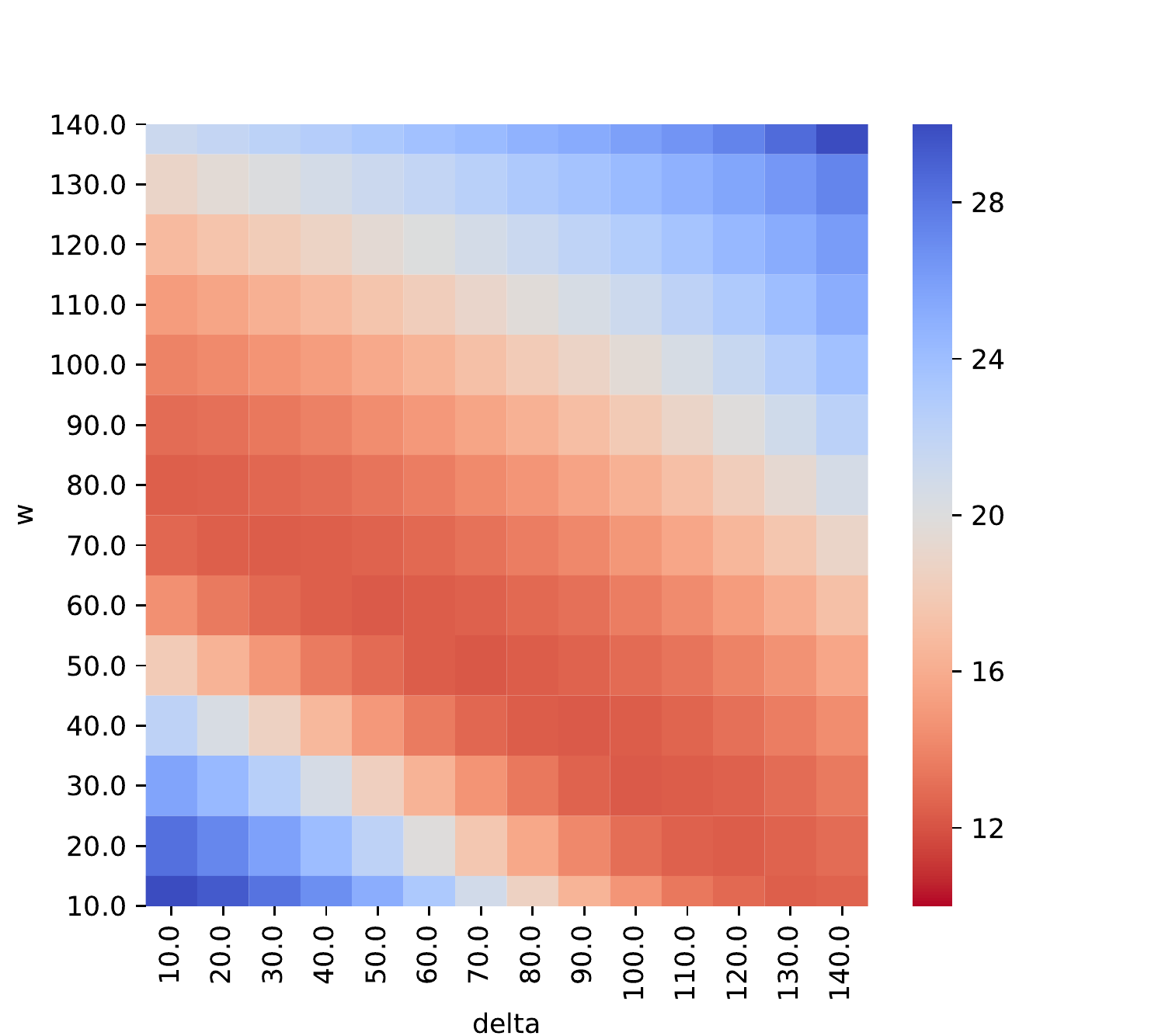}
    \caption{Variation of RMSE with $w$ and $\Delta$ for $S2$}
    \label{subfig:sensitivity_s2}
    \end{subfigure}
        \begin{subfigure}[b]{0.49\linewidth}
        \includegraphics[width = \textwidth]{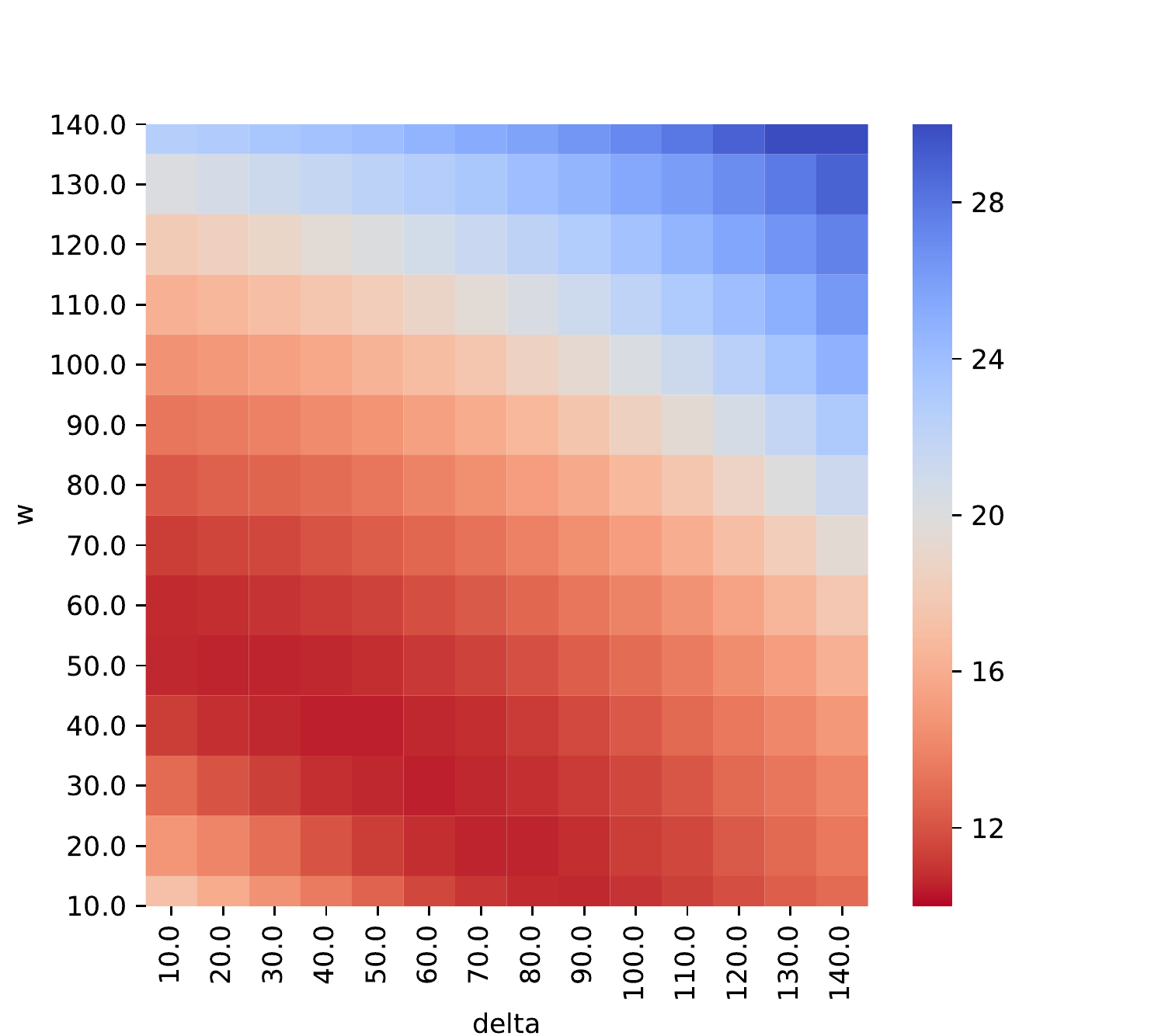}
    \caption{Variation of RMSE with $w$ and $\Delta$ for $S3$}
    \label{subfig:sensitivity_s3}
    \end{subfigure}
    \caption{Heatmap showing RMSE in estimating changepoint in scenarios $S2$ and $S3$ for different values of $w$ and $\Delta$}
    \label{fig:sensitivity_analysis}
\end{figure}
We consider 2 scenarios to study the effect of hyper parameters. Subfigure \ref{subfig:sensitivity_s2} shows the variation of RMSE for scenario $S2$ with continuous jump in mean function and subfigure \ref{subfig:sensitivity_s3} shows the variation of RMSE for scenario $S3$ with smooth jump in mean function (see figure~\ref{fig:cp-mean-curves} for reference). In both the cases, for estimation, we let the order pair $\langle w, \Delta\rangle$ take values in $\{10,20, \dots, 140\}\times \{10,20, \dots, 140\}$. It is evident from figure~\ref{fig:sensitivity_analysis} that there exists a wide range of optimal values of $w$ and $\Delta$, and too high or too low values of both these parameters together are detrimental to the efficacy of the model. 

The difference in the range of values of RMSE between subfigures \ref{subfig:sensitivity_s2} and \ref{subfig:sensitivity_s3} can be attributed to the difference in the mean curve of the data used in the $2$ scenarios. As shown by using a shared color scale, rough-fuzzy CPD is able to achieve a lower RMSE for situation $S3$ than in situation $S2$ with their corresponding optimal hyperparameter values. One possible reason could be the resemblance of the chosen membership function for rough-fuzzy CPD with the change in the mean function for situation $S3$.


\section{Application to Real Data}\label{sec:real-examples}

We consider three real data sets to show the performance of our method rough-fuzzy CPD over the usual regularity measure based methods. Two popular datasets in changepoint analysis, namely ``Flow of the River Nile" data and ``Seatbelts" data regarding monthly Road Casualties in Great Britain $1969-84$, collected from \texttt{datasets} package~\cite{durbin2001time} in \texttt{R}~\cite{R-datasets}. For the third illustration, changepoints are estimated from time varying reproduction rate from coronavirus disease of 2019 (COVID-19) incidence dataset.

\subsection{``Flow of the river Nile" dataset}

The very well-known ``Flow of the river Nile" dataset comprises measurements of the annual flow of Nile river at Aswan from 1871 to 1970. The data has a possible changepoint near 1898 which is associated with the Fashoda incident~\cite{bates1984fashoda} in the same year. 
\begin{figure}[ht]
    \centering
    \begin{subfigure}[b]{0.49\linewidth}
        \includegraphics[width = \textwidth]{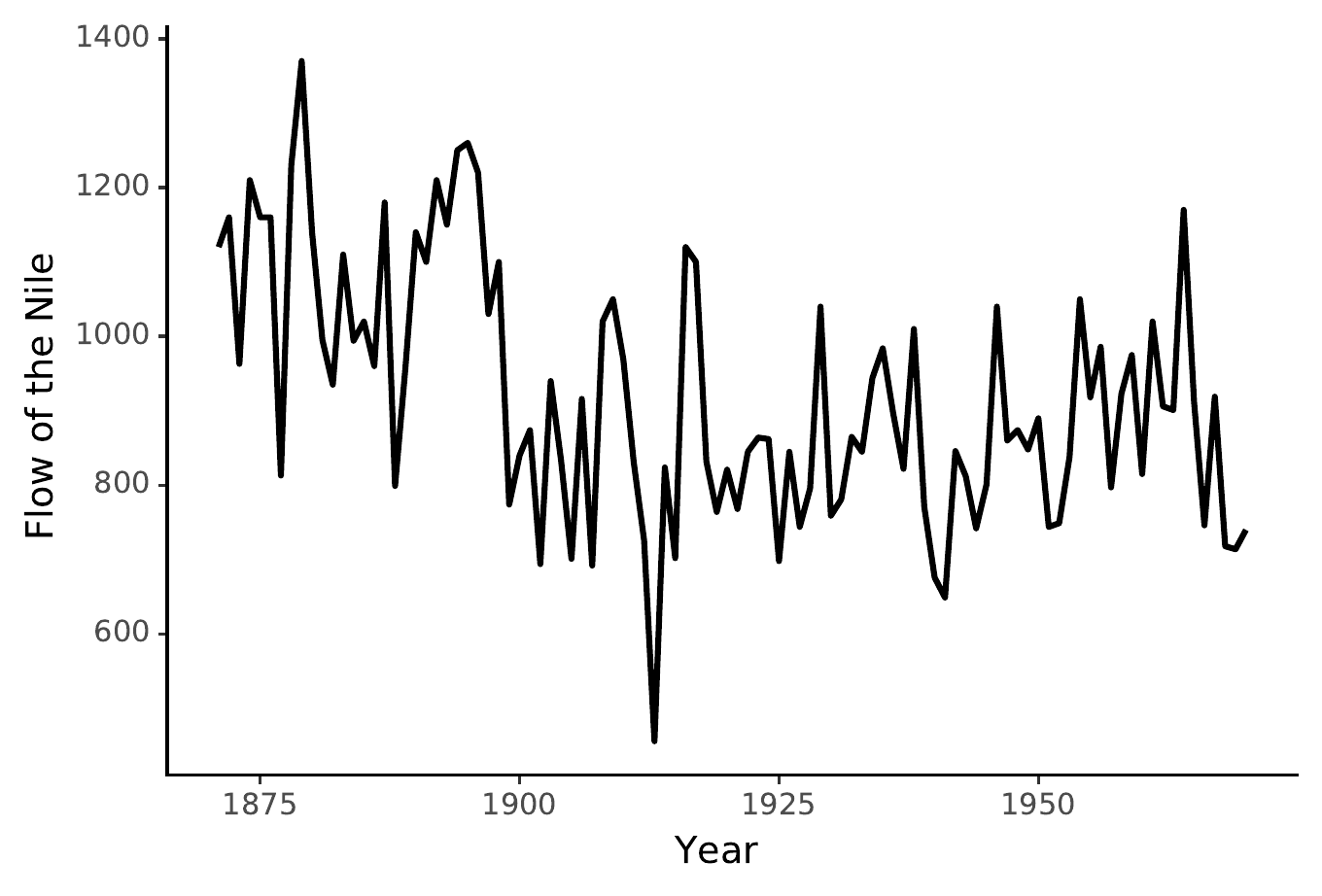}
        \caption{Observed Data}
        \label{subfig:nile_data}
    \end{subfigure}
        \begin{subfigure}[b]{0.49\linewidth}
        \includegraphics[width = \textwidth]{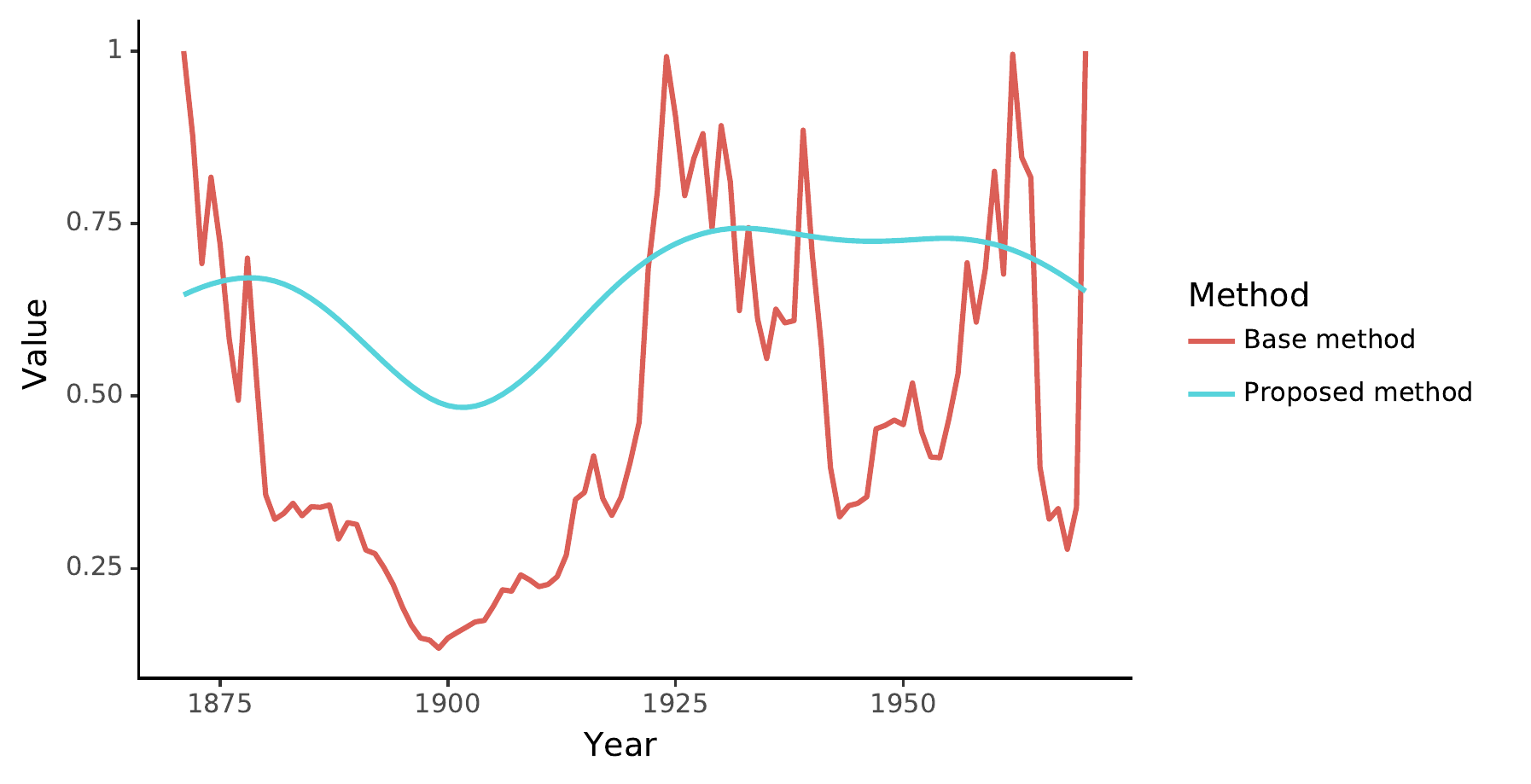}
        \caption{Detected changepoints by base and rough-fuzzy CPD with t-test statistic based regularity measure}
        \label{subfig:nile_result}
    \end{subfigure}
    \caption{Performance of proposed rough-fuzzy CPD and base method (KS-test) on Nile dataset}
    \label{fig:nile}
\end{figure}

The results for Nile dataset is shown in figure~\ref{fig:nile}, where we use t-test as a regularity measure. As shown in subfigure~\ref{subfig:nile_data}, there are multiple minimas detected by the base method, while the proposed rough-fuzzy improvement smooths out these minimas and the false positives are automatically removed. The estimated changepoint by our rough-fuzzy CPD turns out to be at $1902$, which coincides with the completion of Zifta Barrage and Assiut Barrage~\cite{wegmann1922design} and is close to the commonly believed changepoint at 1898.

\subsection{``Seatbelts" dataset }

The popular ``Seatbelts" dataset consists of the number of drivers in Great Britain wearing seatbelts during the time period Jan, $1969$ to Dec, $1984$. There are possibly two evident changepoints present in the data, first one corresponds to the start of seat belt legislation movements from $1972$ when the seatbelts were enforced compulsory in newly manufactured cars, and the second one corresponding to the compulsory enforcement of seat belt wearing while driving in $1983$~\cite{van2020evaluation}.

\begin{figure}[ht]
    \centering
    \begin{subfigure}[b]{0.49\linewidth}
        \includegraphics[width = \textwidth]{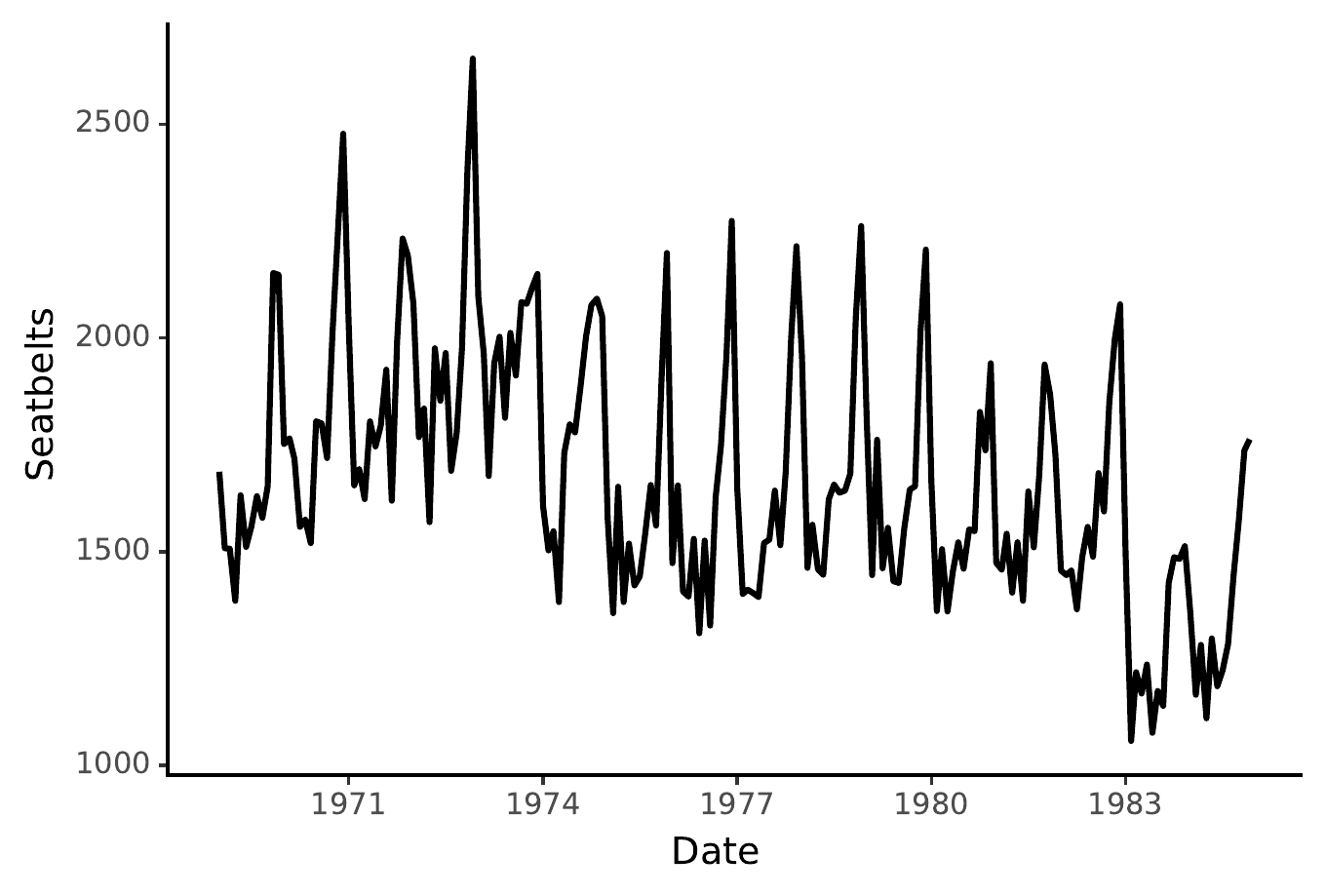}
        \caption{Observed Data}
        \label{subfig:seatbelt_data}
    \end{subfigure}
        \begin{subfigure}[b]{0.49\linewidth}
        \includegraphics[width = \textwidth]{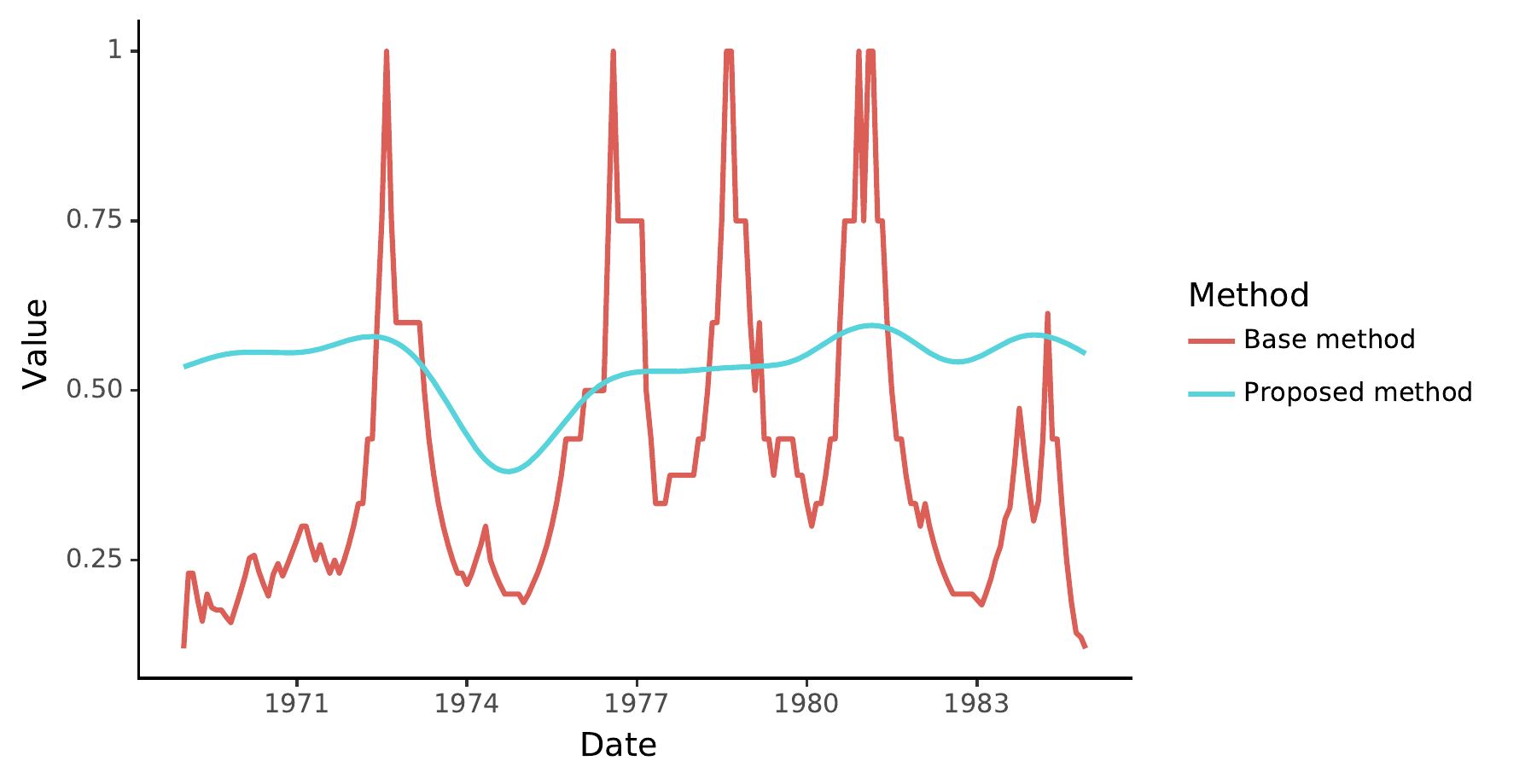}
        \caption{Detected changepoints by base and rough-fuzzy CPD with KS-test statistic based regularity measure}
        \label{subfig:seatbelt_result}
    \end{subfigure}
    \caption{Performance of proposed rough-fuzzy CPD and base method (KS-test) on Seatbelts dataset}
    \label{fig:seatbelt}
\end{figure}

Figure \ref{fig:seatbelt} explains the performance of rough-fuzzy CPD in detecting changepoints for Seatbelts data. Similar to Nile data, the regularity measure based on nonparametric Kolmogorv Smirnov test tends to have too many local minimas, each of which can be possibly thought of as an estimate of changepoint. However, in conjunction with the rough-fuzzy improvement, the prominent local minima appear in $1975$ and $1983$, both of which are in close proximity of the true changepoints.

\subsection{COVID-19 data}
Coronavirus disease of 2019 (COVID-19) is an infectious disease caused by severe acute respiratory syndrome coronavirus 2 (SARS-CoV-2) virus. First detected in Wuhan, China in December, 2019, it has spread across the globe and has infected over 32.2 Million cases by $25^{\text{th}}$ September, 2020. There has been a multitude of epidemiological models trying to estimate the spread of COVID-19 in various countries and states~\cite{Purkayastha2020}. However, many such models fail to predict case counts accurately due to the frequently changing public health measures aimed at simultaneously controlling the spread of the disease and preventing economic crisis. This makes identification of changepoints in the epidemic data of particular importance. Another use of identifying changepoints in spread of COVID-19 is inference regarding the efficacy of non pharmaceutical interventions. 

Such work has been done for some countries like USA~\cite{zhang2020change} and Germany~\cite{Dehningeabb9789}. To identify changepoints, Dehning et al.~\cite{Dehningeabb9789} calculated the time varying rate of transmission $\beta(t)$ in the usual Susceptible Infected Recovered (SIR) model~\cite{kermack1927contribution} and detected changepoints in $\beta(t)$ through the course of the pandemic. For the SIR model, the transmission rate is defined as the rate of transmission of infection from an infected individual to a susceptible individual~\cite{kermack1927contribution}. 

While the SIR model and its parameters are relatively straightforward to understand and easy to estimate, the dynamics of transmission of COVID-19 is much more complex due to the presence of factors like undetected asymptomatic transmissions, latency period, incubation period, delay in reporting, under reporting, different transmission rates for asymptomatic and symptomatic individuals, quarantined and hospitalized individuals, and erroneous testing resulting in false negatives, to name a few. So, $\beta(t)$ alone might not be a suitable representative of the spread of the disease. 

So, we look at another related parameter which is the basic reproduction number $R_0$. While SIR model assumes constant $R_0$ we estimate the time varying reproduction number $R_t$. $R_t$ can be defined as the expected number of cases directly generated by one case in the population at time $t$~\cite{Fraser1557}. We estimate $R_t$ using the \texttt{EpiEstim} package~\cite{10.1093/aje/kwt133} in \texttt{R} based on the incidence curve from $15^{th}$ March to $25^{th}$ September in India. The required data on the number of confirmed cases of COVID-19 has been collected from an API made by volunteer driven covid19india group~\cite{covid19indiaorg2020tracker}. We then apply rough-fuzzy CPD on the estimated $R_t$ to identify the changepoints. Note that modelling the changepoints in a fuzzy manner, instead of crisp, is logical and appropriate here as interventions and measures rolled out by government are impossible to be implemented throughout a vast country like in India instantaneously. 

\begin{figure}[ht]
    \centering
    \includegraphics[width = \linewidth]{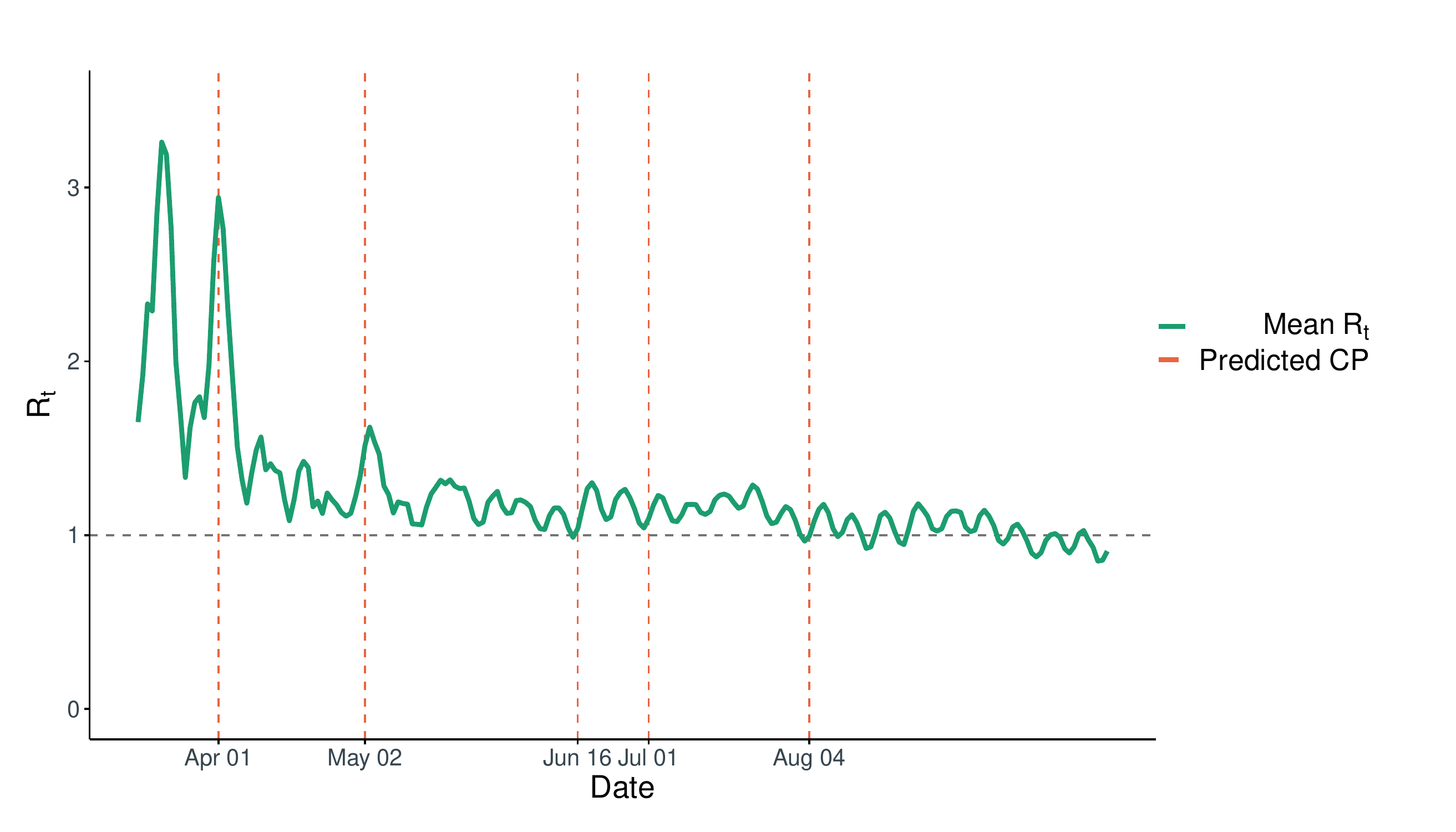}
    \caption{Changepoints in estimated $R_t$ for India}
    \label{fig:R0_india}
\end{figure}

Figure \ref{fig:R0_india} shows the estimated changepoints along with the estimates of $R_t$. We see that our method detects $5$ changepoints on $1^{\text{st}}$ April, $2^{\text{nd}}$ May, $16^{\text{th}}$ June, $1^{\text{st}}$ July and $4^{\text{th}}$ August. Of these, the changepoints on $2^{\text{nd}}$ May, $1^{\text{st}}$ July and $4^{\text{th}}$ August are very close to the date of commencement of Lockdown 3, Unlock 2 and Unlock 3 respectively. The first changepoint (on $1^{\text{st}}$ April) corresponds to the initial elbow in the estimate of $R_t$ after 7 days of the $1^{\text{st}}$ Lockdown in India. This shows that the effect of the $1^{\text{st}}$ Lockdown was not immediate and it took approximately 7 days to initially check the unrestricted spread of the disease. Hence, we infer that the effect of Lockdown 3, Unlock 2 and Unlock 3 were much more pronounced on $R_t$ than the other nationwide interventions in India.

\section{Conclusion}\label{sec:conclusion}

Though gradual changepoints are present in various time series data, they are usually overlooked. Our study aims to contribute in this direction by enriching the existing methodologies using the principles of Fuzzy Rough set theory. In fact, the biggest strength of our approach is that it is independent of the base method, and thus any changepoint detection algorithm expressed as shown in figure~\ref{eqn:old-cp} can be subjected to improvement by our proposal. Here, we have presented only 3 cases of base methods for computing the regularity measure - parametric two sample test, nonparametric two sample tests and stationarity test (unit root test). However, our choices are not limited to the aforementioned cases. Any method for detecting changepoint which provides any type of anomaly score or regularity score can be used in our framework with suitable transformation. Thus our method allows to utilize the rich collection of methods existing for crisp changepoint detection for the problem of fuzzy changepoint detection. 

As various simulations show, even for very simple regularity measures, like t-test statistic and Kolmogorov-Smirnov test statistic, combining them with the rough-fuzzy CPD increases their efficiency by far. In all cases, except for the discrete jump models (where the assumption of fuzziness in changepoint is violated), rough-fuzzy CPD reduces the MSE in detecting the changepoints. In comparison to existing fuzzy and crisp changepoint detection algorithms, rough-fuzzy CPD turns out to be more efficient in estimation of continuous and smooth changes in mean.

The asymptotic distribution mentioned in theorem~\ref{thm:delta-method-single} and theorem~\ref{thm:delta-method-multiple} connects the rough-fuzzy CPD to a hypothesis testing framework, allowing one to output the statistical significance of the estimated changepoints as well. Since, theorem~\ref{thm:delta-method-multiple} allows us to derive the joint asymptotic distribution of the proposed entropy at multiple time points, we are able to rule out false positives and apply the rough-fuzzy CPD for multiple changepoint detection problem. 

From a computational point of view, this paper addresses the issue of obtaining closed form expressions for upper and lower approximations of the rough fuzzy partitions, described by an estimated changepoint, thus allowing the overall algorithm to be much faster. In general, the time and space complexity of the overall algorithm will be dominated by the cost of computing the regularity measures. 

Our approach of using fuzzy and rough set theories has definite advantages which is reflected in the robustness of the rough-fuzzy CPD. Under changes in the signal-to-noise ratio, this method gives a steady improvement over the regularity measure based changepoint detection. We observe that for a diverse range of SNR, the improvement in accuracy of estimates obtained by the rough-fuzzy CPD over that of the base model remains fairly uniform ranging from 88\% to 96\%. On the other hand, increasing the fuzziness $\mathcal{F}$ in the continuous change in mean function increases the efficiency of the rough-fuzzy CPD in comparison to the base statistic used. As fuzziness increases, a rapid increment in efficiency can be observed. Also, rough-fuzzy CPD rarely outputs outlying or spurious false positive estimate of changepoint, precisely because the entropy curve $H^E_{\Delta, \delta, w}(s)$ possesses more smoothness properties than its regularity measure $R(t)$ counterpart. In fact, even in the cases when the base model predicts a bimodal distribution of estimated changepoints, our model shrinks the two modes towards the true changepoint, as illustrated in figure \ref{fig:adf-all}. Finally, through sensitivity analysis, we have shown that the rough-fuzzy CPD performs reasonably well on a wide range of hyperparameter values.

On the flip side, there are some obvious limitations. The rough-fuzzy CPD may not perform good if the assumption of fuzziness in the true changepoint is violated (i.e., there is a discrete jump discontinuity in the mean function), or if the regularity measure $R(t)$ is a bad indicator for the type of changepoint that we are trying to detect. Another limitation could be the choice of the hyperparameters $w, \Delta$ and $\delta$. 

A future possibility for extension of this investigation may consider building a probabilistic view of the detected changepoint, which should enable one to provide a confidence interval for the estimate of changepoint, which is more meaningful in terms of gradual change. Further, rough-fuzzy CPD can be viewed as an extension of an image segmentation algorithm~\cite{image-ambiguity-skpal} into the domain of changepoint detection. A future endeavor in this direction could be to generalize a family of image segmentation algorithms to fit perfectly in the context of a changepoint detection problem.

\section*{Software}

For broader dissemination of our work, we have developed the python package \texttt{roufcp} which is available at \hyperlink{https://pypi.org/project/roufcp/}{pypi.org/project/roufcp/}. All codes for the package has been open sourced and are made available at a github respository \hyperlink{https://github.com/subroy13/roufcp}{github.com/subroy13/roufcp}.

\section*{Acknowledgements}
SK Pal acknowledges the National Science Chair of SERB-DST, Govt. of India, that he is holding currently.


\printbibliography

\end{document}